\documentclass[11pt,reqno]{amsart}  

\usepackage{amscd} 
\usepackage{mathrsfs}
\usepackage{amsmath,amsbsy,amsthm,amsfonts,amssymb,esint}
\usepackage{mathtools}
\usepackage[unicode,psdextra]{hyperref}
\usepackage{graphicx}
\usepackage{tikz}
\usepackage{bm}
\usepackage[normalem]{ulem}
\usepackage{soul}
\usepackage{color}

\usepackage[sort&compress,capitalise,nameinlink]{cleveref}
\crefname{section}{\textsection}{\textsection}
\crefname{subsection}{\textsection}{\textsection}
\crefname{assumption}{Assumption}{Assumptions}
\usepackage{dsfont}
\usepackage{enumerate}

\DeclareMathAlphabet\mathbfcal{OMS}{cmsy}{b}{n}

\newtheorem{theorem}{Theorem}[section]
\newtheorem{lemma}{Lemma}[section]
\newtheorem{proposition}{Proposition}[section]
\newtheorem{corollary}{Corollary}[section] 
\theoremstyle{definition}
\newtheorem{assumption}[]{Assumption}

\newtheorem{remark}{Remark}[section] 
\numberwithin{equation}{section}

\newcommand{\N}{\mathbb{N}}

\newcommand{\Q}{\mathbb{Q}}
\newcommand{\R}{\mathbb{R}}
\newcommand{\C}{\mathbb{C}}

\newcommand{\one}{\mathds{1}}
\newcommand{\xv}{\mathbf{x}}

\newcommand{\yv}{\mathbf{y}}

\newcommand{\diff}{\mathrm{d}}
\newcommand{\eps}{\varepsilon}

\newcommand{\supp}{\mathrm{supp}\,}

\def\AA{\bm{\mathrm{A}}}
\def\VV{V}
\def\H0{H_{0}}
\def\Rz{R_{0}^{\lambda}}
\def\G0{G_{0}^{\lambda}}
\def\K0{K_{0}}
\def\Q0{Q_{0}}
\def\QN{Q_{N}}
\def\GaN{\Xi_{N}^{\lambda}}
\def\HN{H_{N}}
\def\RN{R_{N}^{\lambda}}

\def\UU{U}
\def\muU{\mu_\infty}
\def\nuU{\nu_\infty}

\def\pp{p}

\def\Qinf{Q_{\infty}}
\def\Hinf{H_{\infty}}
\def\gs{\zeta_0}

\def\HH{\mathcal{H}}

\def\Gf0{g_{0}^{\lambda}}
\def\Gt0{g_{0}}

\def\hh{h_{\lambda}}
\def\teA{\vartheta_{\AA}}

\def\Bor{\mathcal{B}}
\def\MeC{\mathscr{M}}

\def\Hil{\mathscr{H}}

\def\dom{\mathscr{D}}
\def\GN{G_{N}^{\lambda}}

\renewcommand{\leq}{\leqslant}
\renewcommand{\geq}{\geqslant}

\newcommand{\lf}{\left}
\newcommand{\ri}{\right}
\newcommand{\disp}{\displaystyle}
\newcommand{\tx}{\textstyle}

\newcommand{\braket}[2]{\lf\langle #1|#2 \ri\rangle}
\newcommand{\braketr}[2]{\lf\langle #1\lf|#2\ri. \ri\rangle}
\newcommand{\braketl}[2]{\lf.\lf\langle #1\ri|#2 \ri\rangle}

\newcommand{\meanlrlr}[3]{\lf\langle #1\lf|#2\ri|#3\ri\rangle}

\newcommand{\beq}{\begin{equation}}
\newcommand{\eeq}{\end{equation}}
\newcommand{\bdm}{\begin{displaymath}}
\newcommand{\edm}{\end{displaymath}}
\newcommand{\bdn}{\begin{eqnarray}}
\newcommand{\edn}{\end{eqnarray}}
\newcommand{\bay}{\begin{array}{c}}
\newcommand{\eay}{\end{array}}
\newcommand{\ben}{\begin{enumerate}}
\newcommand{\een}{\end{enumerate}}
\newcommand{\beqn}{\begin{eqnarray}}
\newcommand{\eeqn}{\end{eqnarray}}
\newcommand{\bml}[1]{\begin{multline} #1 \end{multline}}
\newcommand{\bmln}[1]{\begin{multline*} #1 \end{multline*}}

\title[Homogenization of point interactions]{Homogenization of point interactions}

\author{Domenico Cafiero}
\address{Dipartimento di Matematica, Politecnico di Milano, P.zza Leonardo da Vinci, 32, 20133, Milano, Italy}
\email{domenico.cafiero@polimi.it}

\author{Michele Correggi}
\address{Dipartimento di Matematica, Politecnico di Milano, P.zza Leonardo da Vinci, 32, 20133, Milano, Italy}
\email{michele.correggi@polimi.it}
\urladdr{https://sites.google.com/view/michele-correggi}

\author{Davide Fermi}
\address{Dipartimento di Matematica, Politecnico di Milano, P.zza Leonardo da Vinci, 32, 20133, Milano, Italy\\
and Istituto Nazionale di Fisica Nucleare, Sezione di Milano, Italy}
\email{davide.fermi@polimi.it}
\urladdr{https://fermidavide.com}

\begin{document}

\begin{abstract} 
	We consider a non-relativistic quantum particle in $\mathbb{R}^d$, $d=2$ or $d = 3$, interacting with singular zero-range potentials concentrated on a large collection of points. We analyze the homogenization regime where the intensities of the singular potentials and the distances between the points simultaneously go to zero as their number grows, while the total interaction strength remains finite. Assuming that the singular potentials have negative scattering lengths and are uniformly distributed, we prove the strong resolvent convergence as $N \to \infty$ of the family of operators to a Schr\"{o}dinger operator with a regular electrostatic potential. The result is obtained via  $\Gamma$-converge of the associated quadratic forms. Moreover, in presence of an external trapping potential, the convergence is lifted to uniform resolvent sense.
\end{abstract}


\maketitle

\tableofcontents

\section{Introduction}

\noindent
Zero-range potentials were introduced in the mid `30s to model the scattering of low-energy neutrons by atomic nuclei and to describe related atomic systems \cite{BP35,Th35,Fe36}. In this context, ``potentials'' of Dirac-delta type provide an effective approximation whenever the particle wavelength greatly exceeds the range of very strong forces.
	The first rigorous construction of a point interaction as a singular perturbation of the free Laplacian was given by Berezin and Faddeev  in 1961 \cite{BF61}. Since then the mathematical description of zero-range potentials has been extensively studied within the framework of self-adjoint extensions of symmetric operators, using classical tools such as von Neumann theory, quadratic forms, Krein-type resolvent formulas, and boundary triplets. Notably, such singular interactions can be realized as renormalized limits of resonant short-range potentials. For comprehensive expositions, we refer to the monographs \cite{AGHKH88,AK00} and the references therein.

	In this work, we study the quantum dynamics of a non-relativistic particle moving in a smooth electromagnetic background and interacting with several point-like scatterers. Formally, the system is described by a Schrödinger operator of the form
    \begin{equation}	\label{eq:formalHN}
		(- i \nabla + \AA)^2 + \VV + \mbox{``}\tx\sum_{j=1}^N \gamma_j\,\delta_{\xv_j}\,\mbox{''},
    \end{equation}
where $\AA$ is a smooth vector potential and $\VV$ a regular electrostatic potential. The points $ \xv_j \in \R^d$ identify the positions of the scatterers, although the expression above is purely formal (see below) and the zero-range perturbation should not be thought of as a potential. Our aim is to analyze the homogenization (or mean-field) regime in which the number of scatterers $N$ diverges, while they cluster within a fixed space region and their interaction strengths scale proportionately. In this limit, the collective effect of the scatterers results in an additional electrostatic potential, determined by the distribution and individual strengths of the centers.
	
	This type of problem in three dimensions was first addressed in \cite{FHT88} (see also \cite{FT92,FT93} and references therein), where the authors treated the point interactions as independent and identically distributed random variables. Under a suitable scaling, they showed that the associated family of Hamiltonians converge in probability, in strong resolvent sense, to a Schr\"odinger operator with an additional electrostatic potential, and they also provided quantitative estimates for the fluctuations around this limit. Their approach relies on explicit resolvent formulas, combined with earlier homogenization results for the Laplacian on domains with many small holes (see, e.g., \cite{RT75,PV80,Oz83,FOT85,FOT87}). For related works in one or more dimensions we refer to \cite{BFT98} and \cite{EN03,Oz06}, respectively. A somewhat similar model is the so-called Rayleigh gas, which is supposed to describe the low-energy zero-range interaction of a quantum particle with $ N $ scattering centers \cite{DFT,CDF}, which are harmonically bounded to their equilibrium positions. Let us also mention that the approach originally developed in \cite{FHT88} was applied more recently in \cite{BCT18} to derive a probabilistic characterization of the effective dynamics of a quantum particle in a quantum Lorentz gas with Gross–Pitaevskii-type interactions. 
	
	Explicit Krein-type resolvent formulas play a crucial role in all the works mentioned above. In contrast, here we develop an approach based on quadratic forms and $\Gamma$-convergence techniques, using standard results to deduce strong resolvent convergence as a direct consequence. We refer to \cite{Br02,DM93} for classical references on $\Gamma$-convergence. Our framework is arguably more flexible, allowing the inclusion of smooth electromagnetic background fields and a unified treatment of models in dimensions $d = 2$ and $d = 3$. The one-dimensional case could also be included, but we omit it since it is qualitatively much simpler (see, e.g., \cite{BFT98}). Furthermore, our assumptions seem to be more natural, being closer in spirit to those in \cite{FHT88}, and avoiding some of the technical requests introduced in \cite{FHT88,EN03,Oz06}. We also stress that $\Gamma$-convergence allow us to upgrade the convergence to uniform resolvent sense if the original system has discrete spectrum, which occurs for instance if the external electromagnetic field induces trapping. However, as a drawback we have to restrict to point interactions with negative scattering lengths to ensure equi-coerciveness of the quadratic forms.
	
\medskip
    
    The paper is organized as follows. 
    In \cref{sec:main}, after reporting the main definitions and fixing the notation, we present the precise assumtions and outline a heuristic derivation of the limiting operator. Therein, we also state our main rigorous results. 
    \cref{sec:prel} collects some preliminaries on $\Gamma$-convergence, complex measures, integral kernels and Hilbert scales.
    In \cref{sec:proofs}, we treat separately the proofs of the $\Gamma-\liminf$ and $ \Gamma-\limsup $, and the proof of uniform resolvent convergence under the trapping assumption.
\bigskip

\begin{footnotesize}
\noindent
\textbf{Acknowledgments.}
The present research has been supported by the {\it MUR} grant ``Dipartimento di Eccellenza 2023-2027'' of Dipartimento di Matematica, Politecnico di Milano and 
{\it PRIN 2022} grant ``ONES - OpeN and Effective Quantum Systems'' (n. 2022L45WA3). DF acknowledges the support of
INdAM-GNFM through {\it Progetto Giovani 2020} ``Emergent Features in Quantum Bosonic Theories and Semiclassical Analysis''.
MC acknowledges the support of {\it PNRR Italia Domani} and {\it Next Generation Eu} through the {\it ICSC National Research Centre for High Performance Computing, Big Data and Quantum Computing}. DC and DF are grateful to Alessandro Teta for insightful and stimulating discussions.
\end{footnotesize}

\section{Setting and Main Results}\label{sec:main}

Let us first consider the Schr\"odinger operator
	\begin{equation}\label{eq:H0def}
		\H0 := (- i \nabla + \AA)^2 + \VV ,
	\end{equation}
where $\AA \in C^{\infty}(\R^d;\R^d)$ grows at most polynomially at $ \infty $, and $\VV \in L_{\mathrm{loc}}^{p}(\R^d;\R^{+})$ for some $ p > 3 $ is non-negative. These assumptions are certainly not optimal and there is room for improvement (see \cref{rem: assumptions} below), but we stick to them for the sake of concreteness. Note, in particular, that the cases of uniform external magnetic fields and polynomial trapping electric potentials are included. We also remark that $\H0$ is a self-adjoint operator on $L^2(\R^d)$, with domain
	\begin{equation*}
		\dom(\H0) := \lf\{ \psi \!\in\! L^2(\R^d) \;\big|\; (- i \nabla \!+\! \AA)^2 \psi, \VV \psi  \in\! L^2(\R^d) \ri\} .
	\end{equation*}
It is evident that $\H0$ is positive definite, thus in particular $\sigma(\H0) \subset [0,\infty)$. Taking this into account, in the sequel we will refer to the associated resolvent operator, i.e., for $\lambda > 0$,
	\begin{equation*}
		\Rz := (\H0 +\lambda^2)^{-1} : L^2(\R^d) \to \dom(\H0),
	\end{equation*}
	and to its integral kernel (a.k.a. Green function) $\G0(\xv,\yv)$.
Let us also mention that the quadratic form corresponding to the Hamiltonian $H_0$ is given by
	\begin{eqnarray*}
		\dom[\Q0] & := & \lf\{\psi \!\in\! L^2(\R^d) \;\big|\; (- i \nabla \!+\! \AA) \psi\!, \sqrt{\VV} \psi \in\! L^2(\R^d) \ri\} ,\\
		\Q0[\psi] & := & \big\|(-i \nabla \!+\! \AA) \psi\big\|_2^2 + \big\|\sqrt{V}\psi\big\|_2^2\,,
	\end{eqnarray*}
where $ \lf\|\cdot \ri\|_2 : = \lf\|\cdot \ri\|_{L^2(\R^d)} $ for short.

A proper self-adjoint operator in $L^2(\R^d)$ matching the heuristic expression \eqref{eq:formalHN} can be rigorously defined as a singular perturbation of $\H0$, {\it e.g.}, via the construction of its quadratic form, see \cite{Te90}: for any fixed $\lambda>0$, we set
	\begin{eqnarray}
		\dom[\QN] & := & \lf\{\psi = \phi_\lambda + \tx\sum_{j=1}^N q_j\, \G0(\,\cdot\,, \xv_j) \in L^2(\R^d) \;\big|\; \phi_{\lambda}\!\in\! \dom[\Q0],\;\, q_1,...\,,q_N \!\in\! \C \ri\} , \nonumber  \\
		\QN[\psi] & := & \Q0 \lf[\phi_{\lambda}\ri] + \lambda^2 \lf\|\phi_{\lambda} \ri\|_{2}^2 - \lambda^2 \lf\|\psi\ri\|_{2}^{2} + \tx\sum_{i,j=1}^N  \overline{q_i}\, \big[\GaN\big]_{ij}\, q_j\,, \label{eq:QNdef}
	\end{eqnarray}
where $\overline{q_i}$ is the complex conjugate of $q_i$ and $\GaN$ is the $N\times N$ Hermitian matrix with entries
	\begin{equation}\label{eq:GaN}
		\big[\GaN\big]_{ij} \!:=\! \left(\alpha_j + \lim_{\xv \to \xv_j}\! \big( G_0^{\lambda_0}(\xv,\xv_j) - \G0(\xv,\xv_j) \big)\!\right) \delta_{ij} - \G0(\xv_i,\xv_j)\,(1-\delta_{ij})\,, \quad \mbox{for }i,j \!\in\! \{1,...\,,N\}\,.
	\end{equation}
Here, $\lambda_0 > 0$ is a fixed reference parameter and $\alpha_j \in \R$ are finite constants labeling the form\footnote{For $d = 3$ and zero electromagnetic potentials, the constants $\alpha_j$ and the coefficients $\gamma_j$ in \eqref{eq:formalHN} satisfy the formal identity $\gamma_j = -c\,\varepsilon + \alpha_j \varepsilon^2$, for a universal constant $ c > 0 $ and for $\varepsilon \to 0$, see \cite{BF61} and \cite[\S II.1.2 and Appendix H]{AGHKH88}, {\it i.e.}, the coefficients $ \alpha_j $ are obtained through a suitable renormalization procedure from the ``bare'' parameters $ \gamma_j $. Similar relations hold for $d = 2$ and non-zero electromagnetic fields.} $\QN$. 
We note that the coefficients $\alpha_j$ introduced here differ from those used in \cite{AGHKH88} by a finite additive constant, depending on $\lambda_0$.
By direct generalization of the arguments outlined in \cite{Te90}, it can be inferred that the quadratic form $\QN$ is independent of $\lambda>0$, closed and bounded from below, for both $d = 2$ and $d = 3$. Therefore, it uniquely identifies a self-adjoint operator in $L^2(\R^d)$. This is given by
    	\begin{eqnarray}
		&& \hspace{-1cm}\dom(\HN) =  \lf\{\psi=\phi_\lambda+ \tx\sum_{j=1}^N q_j\, \G0(\,\cdot\,, \xv_j) \in L^2(\R^d) \;\,\big|\; \phi_{\lambda}\!\in\! \dom[\H0],\;\, q_1,...\,,q_N \!\in\! \C\,,\;\, \mbox{ and} \ri. \nonumber \\
		&& \hspace{7cm} \lf. \phi_{\lambda}(\xv_i) = \tx\sum_{j=1}^N\! \big[\GaN\big]_{ij}\, q_j\; \mbox{ for }\, i = 1,...\,,N \ri\} interactio, \nonumber  \\
		 && \hspace{-1cm}\lf(\HN + \lambda^2\ri) \psi =  \lf(\H0 + \lambda^2\ri) \phi_{\lambda}\,. \label{eq:HNdef}
	\end{eqnarray}
The resolvent $\RN := (\HN +\lambda^2)^{-1} : L^2(\R^2) \to \dom(\HN) $ is an integral operator with kernel (see \cite[\S II.1.1, Eq. (1.1.33)]{AGHKH88} and \cite{Po01})
	\begin{equation*}
		\GN(\xv,\yv) = \G0(\xv,\yv) + \tx\sum_{i,j=1}^N \big[\GaN\big]^{-1}_{ij}\, \G0(\xv, \xv_i)\,\G0(\xv_j, \yv)\,.
	\end{equation*}	
As an aside, let us also mention that, upon varying the coefficients $\alpha_j \in \R \cup \{\infty\}$, the corresponding family of Hamiltonians $\HN$ comprises all the so-called {\em local} self-adjoint extensions of the (closable) symmetric operator $\H0 \!\upharpoonright\! C^{\infty}_c(\R^d \setminus \{\xv_1,...\,,\xv_N\})$.

	\begin{remark}[Regularity of $\AA$ and $\VV$]
		\label{rem: assumptions}
		The stated hypotheses on the external potentials $\AA$ and $\VV$ are not meant to be optimal. The main reason why we assume such a regularity on the external potentials is to control the behavior of the Green function $ G_0^\lambda(\xv,\yv) $ on the diagonal, exploiting results proven in \cite{BK13,BGP07} (see \cref{lemma:G0Gfineq,lemma: G regularity} below), which presumably hold also for less regular $ \AA $ and $ \VV $. However, a throughout analysis of the minimal regularity conditions lies beyond the scope of this work.	
		Note also that the positivity assumption on $ \VV $ can be relaxed and  boundedness from below is certainly enough (it suffices to shift the whole quadratic form of a positive energy), but Kato-smallness of the negative part of $ \VV $ is also fine. 
	On the other hand, let us emphasize that the behavior of $\AA,\VV$ at infinity is unconstrained. In particular, these potentials are allowed to diverge at infinity, which may result in a trapping mechanism for the particle whose dynamics is governed by $\H0$.
	\end{remark}

	\begin{remark}[Well-posedness of $\QN$ and $\HN$]\label{rem:wellQNHN}
		We first notice that, for any fixed $ \yv \in \R^d$, the function $ \xv \mapsto \G0(\xv,\yv)$ is indeed square-integrable on $\R^d$, $d \leqslant 3$ \cite[Thm. 19]{BGP07}. Moreover, as a consequence of the second resolvent identity, the difference $G_0^{\lambda_0}(\xv,\yv) - \G0(\xv,\yv)$ is jointly continuous in $(\xv,\yv) \in \R^d \times \R^d$, again if $d \leqslant 3$ \cite[Lemma 9]{BGP05}. This ensures that the limit for $ \xv \to \xv_j$ appearing in the definition \eqref{eq:GaN} of $\GaN$ exists and is finite. We stress that this limit does not depend on $N$.
		On the other hand, in view of the basic identity $\G0(\xv,\yv) = \overline{\G0(\yv,\xv)}$, it is easy to check that the matrix $\GaN$ is Hermitian; in particular, the diagonal limit mentioned above is real-valued. Accordingly, the form $\QN$ is also real-valued. As an aside, let us mention that when $\AA = \bm{0}$ (no magnetic field), the Hamiltonian  $\H0$ commutes with complex conjugation; in this case $\G0(\xv,\yv)$ is real valued and thus symmetric, {\em i.e.}, $\G0(\xv,\yv) = \G0(\yv,\xv)$. Further properties of $\G0(\xv,\yv)$ will be discussed in \cref{subsec:G0prop}.
		Finally, since $\dom(\H0) \subset H^2_{\mathrm{loc}}(\R^d)$, and hence $\dom(\H0) \subset C^0(\R^d)$ by Sobolev embedding \cite[Prop. 2.13]{CP82}, the pointwise evaluation $\phi_\lambda(\xv_i)$ appearing in the definition \eqref{eq:HNdef} of $\dom(\HN)$ is well defined.
	\end{remark}

Our goal in this work is to analyze the limiting behavior of $\HN$ as $N\to \infty$ in a mean-field regime, where the interaction centers accumulate within a fixed and bounded spatial region, while their individual intensities weaken proportionally, so as to keep the total strength finite. 
More precisely, we shall henceforth refer to the following assumptions about the positions of the points $\{ \xv_j\}_{j \in \N}$ and the parameters $\{\alpha_j\}_{j \in \N}$.

	\begin{assumption}[Distribution of points]\label{ass:UU}\mbox{}	\\
         There exists a non-negative function $\UU \in L^\infty(\R^d)$, with compact support and $\|\UU\|_{L^1}=1$, such that $\xv_j \in \supp \UU$ for all $j \in \N$ and\footnote{Here, we denote by $ \mathrm{w} $ the weak convergence of measures, {\it i.e.}, the convergence in the dual of continuous functions with support containing $ \supp(U) $ (see \cref{subsec:meas}).}
         \begin{equation}\label{eq:Vwconv}
             \mu_N := \frac{1}{N}\sum_{j=1}^N\delta_{\xv_j} \xrightarrow[N \to \infty]{\mathrm{w}} \muU\,,
             \qquad \mbox{with }\, \diff \muU(\xv) = \UU(\xv) \, \diff \xv\,.
         \end{equation}
    \end{assumption}

	\begin{assumption}[Minimal distance]\label{ass:xixj}\mbox{}	\\
        There exists a constant $\ell >0$, independent on $N$, such that
        \begin{equation}\label{eq:infxixj}
            \inf_{1 \,\leqslant\, i \,<\, j \,\leqslant\, N} |\xv_i - \xv_j |\geq \ell\, N^{-1/d}\,, \qquad
            \forall N \in \N\,.
        \end{equation}
    \end{assumption}

	\begin{assumption}[Interaction strengths]\label{ass:a}\mbox{}	\\
		There exists a bounded positive function $a \in C^{0}(\R^d)$ 	such that
		\begin{equation}
			\alpha_j = N\,a(\xv_j)\,, \qquad  \forall N \in \N \mbox{ and } \forall j \in \{1,...\,,N\}\,.
		\end{equation}
    \end{assumption}

	\begin{figure}[t!]\label{fig:pic}
		\includegraphics[width=10cm]{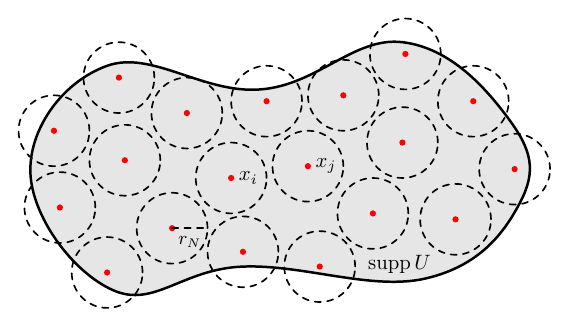}
		\label{fig:pict}
		\caption{A schematic illustration of \cref{ass:UU} and \cref{ass:xixj}. The points $ \xv_j$ cluster within a fixed region following the density distribution $\UU$, while maintaining a minimal mutual distance $r >\ell\,N^{-1/d}$, see \eqref{eq:infxixj}.}
	\end{figure}

	\begin{remark}[Assumptions]
		\cref{ass:UU} indicates that $\UU$ serves as the limit density distribution for the points at which the singular potentials are concentrated. In particular, the support of $\UU$ identifies the region of space where the points accumulate.
		On the other hand, \cref{ass:xixj} constrains the rate at which the points may cluster. This condition is crucial to guarantee the equi-coerciveness of the quadratic forms, a key element to establish $\Gamma$-convergence.
		The latter \cref{ass:a} prescribes a specific scaling for the interaction strengths. Considering that the parameters $\alpha_j$ are inversely related to the scattering lengths of the individual centers \cite[\S I.1.4]{AGHKH88}, this assumption actually implies that each single-center interaction vanishes as $N \to +\infty$. Indeed, it is easy to see that the choice $ \alpha_j = +\infty $, $ \forall j $, in \eqref{eq:GaN} identifies the unperturbed operator $ H_0 $.
	\end{remark}
     
     Before stating our main results, let us heuristically describe the mechanism of convergence. The rescaling of the parameters $ \lf\{ \alpha_j \ri\}_ {j \in \{1,...\,,N\}} $ as in \cref{ass:a} implies that the charges $ \lf\{ q_j \ri\}_ {j \in \{1,...\,,N\}} $ also scales with $ N $ as
	\begin{equation}\label{eq:qjres}
		q_j = \tfrac{1}{N}\,\pp_j \,, \qquad \mbox{for all } j = 1,...\,,N ,
	\end{equation}
for some complex coefficients $\pp_j$ of order 1. Suppose now that there exists a sufficiently regular function $p : \R^2 \to \C$, which is independent of $ N $ and such that
	\begin{equation}\label{eq:pjpx}
		\pp_j = p(\xv_j) \,, \qquad \mbox{for all } j = 1,...\,,N \,.
	\end{equation}
Then, by a Riemann sum approximation, the last term in the expression \eqref{eq:QNdef} becomes
	\bmln{
		\frac{1}{N^2} \sum_{i,j=1}^N  \overline{\pp_i}\, \big[\GaN\big]_{ij}\, \pp_j
		= \frac{1}{N^2} \sum_{i,j=1}^N  \overline{p(x_i)} \,\big[ \left(N a(\xv_j) + \mathcal{O}(1) \right)\delta_{ij} - \G0(\xv_i,\xv_j)\,(1-\delta_{ij}) \big]\, p(\xv_j) \\
		= \frac{1}{N} \sum_{j=1}^N a(\xv_j)\, |p(\xv_j)|^2 - \frac{1}{N^2} \sum_{i \neq j} \, \overline{p(\xv_i)} \,\G0(\xv_i,\xv_j)\, p(\xv_j) + o(1) \\
		\xrightarrow[N \,\to \,\infty]{} \int_{\R^d} a(\xv)\,|p(\xv)|^2\, \UU(\xv)\, \diff \xv - \int_{\R^d \times \R^d} \overline{p(\xv)}\,\G0(\xv,\yv)\,p(\yv)\, \UU(\xv)\,\UU(\yv)\, \diff \xv\,\diff \yv\,.
	}
From here, assuming that $\phi_{N,\lambda} \to \phi_{\infty,\lambda}$ in the appropriate topology, we infer that
	\bml{
		\QN[\psi_N] + \lambda^2 \lf\|\psi_N \ri\|_2^2 \; \xrightarrow[N \,\to \,\infty]{}\; 
		\lf\| (-i \nabla + \AA) \phi_{\infty,\lambda} \ri\|_{2}^{2} + \big\|\sqrt{V} \phi_{\infty,\lambda} \big\|_2^2 + \lambda^2 \, \lf\|\phi_{\infty,\lambda} \ri\|_{2}^2 \\
		 + \braketr{p}{a\,\UU p} - \braketr{\UU p}{\Rz\, \UU p} .\label{eq:QNeuri}
	}
	Furthermore, exploiting the boundary condition \eqref{eq:HNdef} in the operator domain $\dom(\HN)$, we get
	\bmln{
		\phi_{N,\lambda}(\xv_i) 
		= \frac{1}{N} \sum_{j=1}^N \lf[ \lf(N a(\xv_j) + \mathcal{O}(1) \ri)\delta_{ij} - \G0(\xv - \xv_j)\,(1-\delta_{ij}) \ri]\, p(\xv_j)  \xrightarrow[N \to\infty]{}\;  a(\xv_i)\, p(\xv_i)\\
		 - \int_{\R^d} \G0(\xv_i,\yv)\, p(\yv)\,\UU(\yv)\, \diff \yv \,,
	}
which ultimately suggests that
	\begin{equation}\label{eq:psiap}
		\psi_\infty(\xv) = a(\xv)\,p(\xv)\,.
	\end{equation}
	Hence, recalling \eqref{eq:HNdef},
	\beq\label{eq:psiinf}
		\psi_N(\xv) = \phi_{N,\lambda}(\xv) + \frac{1}{N} \sum_{j=1}^N p(\xv_j)\, \G0(\xv, \xv_j)
		\;\xrightarrow[N \,\to \,\infty]{}\; \phi_{\infty,\lambda} + \int_{\R^d} \G0(\xv,\yv) \, p(\yv)\,\UU(\yv)\, \diff \yv \,,
	\eeq
	so that
	\begin{align*}
		\meanlrlr{\psi_\infty}{\H0 \!+\! \lambda^2}{\psi_\infty}
		& = \Q0[\phi_{\infty,\lambda}] + \lambda^2 \lf\|\phi_{\infty,\lambda}\ri\|_{2}^2
			+ 2\Re \braketr{\phi_{\infty,\lambda}}{\UU p} + \braketl{ \Rz\, \UU p}{\UU p}	\\
		& = \Q0[ \phi_{\infty,\lambda}] + \lambda^2 \lf\|\phi_{\infty,\lambda} \ri\|_{2}^2
			+ \braketr{ p }{a\,\UU p} - \braketr{\UU p }{\Rz\, \UU p}  + \braketr{\psi_\infty}{\tfrac{\UU}{a} \,\psi_\infty}\,,
	\end{align*}
where we have exploited the fact that 
\bdm
	\braketr{\psi_{\infty}}{\UU p} = \braketr{\psi_{\infty}}{\tfrac{\UU}{a}\, \psi_\infty} = \braketr{p}{a\, \UU p}.
\edm
Summing up, 
	\begin{align}
		\meanlrlr{\psi_\infty}{\H0 - \tfrac{\UU}{a}}{\psi_\infty}
		+ \lambda^2\, \lf\|\psi_\infty \ri\|_2^2 
		& = \Q0[ \phi_{\lambda,\infty}] + \lambda^2 \lf\|\phi_{\infty,\lambda} \ri\|_{2}^2
			+ \braketr{p}{a\, \UU p} - \braketr{\UU p}{\Rz\, \UU p} , \label{eq:Qinfeuri}
	\end{align}
	which yields the expession of the limiting quadratic form by comparison of \eqref{eq:QNeuri} with \eqref{eq:Qinfeuri}:
	\begin{equation}
		\Qinf[\psi_\infty] = \Q0[\psi_\infty] - \meanlrlr{\psi_\infty}{\tfrac{\UU}{a}}{\psi_\infty}, \label{eq:Qinf}
	\end{equation}
where $\UU$ and $a$ are the functions introduced previously in \cref{ass:UU} and \cref{ass:xixj}. Note that, under these assumptions, the ratio $\UU/a$ belong to $L^\infty(\R^d)$, which ensures that, if the guess if correct, $\Qinf$ is a Kato-small perturbation of the initial quadratic form $\Q0$ and $ \dom[\Qinf] = \dom[\Q0] $. Furthermore, $\Qinf$ is closed and bounded from below, and therefore it uniquely identifies a self-adjoint operator, namely,
	\begin{equation}\label{eq:Hinf}
		\Hinf = \H0 - \tfrac{\UU}{a}\,, \qquad 
		\dom(\Hinf) = \dom(\H0)\,.
	\end{equation}

We are now ready to state our main result. Here and in the sequel, we understand that the quadratic forms $\QN$ and $\Qinf$ are extended to the entire Hilbert space $L^2(\R^d)$ by setting them equal to $+\infty$ outside of their respective domains of definition.

    \begin{theorem}[$\Gamma-$convergence]	\label{thm:main}
    		\mbox{}	\\
		Let \cref{ass:UU,ass:xixj,ass:a} hold. Then, the sequence of quadratic forms $\lf\{ \QN \ri\}_{N \in \N}$ $\Gamma$-converges to $\Qinf$ with respect to both the weak and strong topology of $L^2(\R^d)$. More precisely, the following two conditions are fulfilled:
		\begin{enumerate}[(i)] 
			\item $(\Gamma-\liminf)$ for any sequence $ \lf\{\psi_N \ri\}_{N \in \N} \subset L^2(\R^d)$ such that $\psi_N \xrightarrow[N \to +\infty]{\mathrm{w}} \psi_\infty \in L^2(\R^d) $, there holds
				\begin{equation}\label{eq:liminf}
					\Qinf[\psi_{\infty}] \leq \liminf_{N \,\to\, +\infty} \QN[\psi_N]\,;
				\end{equation}
			\item $(\Gamma-\limsup)$ for any $\psi_\infty \in \dom[\Qinf]$, there exists a sequence $ \lf\{\psi_N\ri\}_{N\in\N} \subset L^2(\R^d)$, with $\psi_N\in\dom[\QN]$ for all $N \in \N$, such that $\psi_N \xrightarrow[N \to +\infty]{} \psi_\infty \in L^2(\R^d)$ and
				\begin{equation}\label{eq:limsup}
					\Qinf[\psi_{\infty}] \geq \limsup_{N \,\to\,+\infty} \QN[\psi_N]\,.
				\end{equation}
		\end{enumerate}				
	\end{theorem}

As a direct consequence of \cref{thm:main}, by standard arguments of $\Gamma$-convergence theory, we obtain the following result regarding the family of self-adjoint operators $ \lf\{\HN \ri\}_{N \in \N}$ and the limit Hamiltonian $\Hinf$.

	\begin{corollary}[Operator convergence]	\label{cor:main}	\mbox{}	\\
		Let \cref{ass:UU,ass:xixj,ass:a} hold. Then, the sequence $\{\HN\}_{N \in \N}$ converges to $\Hinf$ in  strong resolvent sense as $N \to \infty$.
		Moreover, if $\H0$ has compact resolvent, then $\{\HN\}_{N \in \N}$ converges to $\Hinf$ in  norm resolvent sense.
	\end{corollary}

	\begin{remark}[Bound states and dynamics]
		Since $\UU$ and $a$ are both non-negative, the limiting Hamiltonian $\Hinf$ comprises an attractive electrostatic potential which may support negative-energy bound states. It is worth noting that also the Hamiltonians $\HN$ typically possess a non-empty discrete spectrum below the continuous threshold \cite[\S II.1.1, Thm. 1.1.4]{AGHKH88}. The result of $\Gamma$-convergence established in \cref{thm:main} guarantees any sequence of ground states of $\{\HN\}_{N \in \N}$ converges to the ground state of $\Hinf$ \cite[Chpt.\! 7]{DM93}. On top of that, when $\H0$ has compact resolvent, the uniform resolvent convergence stated in \cref{cor:main} implies that the purely discrete spectrum of $\H0$ actually consists of eigenvalues which arise as accumulation points of sequences of eigenvalues of $\HN$ \cite[Prop.\! 10.2.4 and Cor. 10.2.5]{DeO09}.
		Let us finally recall that strong resolvent convergence is in fact equivalent to strong convergence of the unitary operators $e^{-i t \HN}$ to $e^{-i t \Hinf}$, for any $t \in \R$, and thus it also provides the convergence of the corresponding quantum dynamics \cite[Prop. 10.1.8]{DeO09}.
	\end{remark}
	
	\begin{remark}[Relaxing the assumptions]
		Let us briefly comment on possible relaxations of \cref{ass:UU,ass:xixj,ass:a}. On the one hand, the theory developed here is stable under mild violations of these hypotheses. The same conclusions would hold if a small number of interactions, say of order $\mathcal{O}(1)$ as $N\to\infty$, did not satisfy the stated assumptions.
		Likewise, the compact-support requirement on the density $\UU$ could be replaced by a sufficiently fast decay at infinity, at the price of some additional effort. Subleading corrections in \cref{ass:a} could also be incorporated; for instance, one may allow $\alpha_j = N a(\xv_j) + \mathcal{O}(1)$ as $N\to\infty$.
		On the other side, it would pose major challenges to consider a large number of non-positive interaction parameters $\alpha_j \leqslant 0$ or to permit the centers to accumulate faster than allowed by the lower bound \eqref{eq:infxixj} on the mutual distances. In these cases, the equi-coerciveness of the quadratic forms $\QN$ would fail. As a result, the behavior of the system could become significantly different, potentially exhibiting new phenomena that are not captured by the current analysis.
	\end{remark}

	\begin{remark}[Domains with boundaries and curved geometries] 
	Our arguments could be extend with minor effort to flat spatial domains with boundaries. More precisely, the proofs remain essentially unchanged if, instead of working on $\R^d$, one were to consider a domain $\Omega$, bounded or unbounded, with a sufficiently regular boundary $\partial \Omega$.
	The only structural requirement is that the differential operator $\H0$ admits a self-adjoint realization on $L^2(\Omega)$, which in turn relies on specifying appropriate boundary conditions on $\partial\Omega$. To avoid technical complications, one should perhaps ensure that the points where the singular potentials are supported do not accumulate near the boundary.
	In the same spirit, one could even extend our results to analogous models living on curved Riemannian manifolds, with the obvious adjustments to the definitions of the basic quadratic forms and related Hamiltonian operators.
	\end{remark}

\section{Preliminaries}\label{sec:prel}

    In this section, we recall some basic facts on $\Gamma$-convergence, complex measures, integral kernels and Hilbert scales, which will be used in the proofs to be presented in the next section. For more details on these topics we refer to \cite{Br02,DM93}, \cite{Bo18,Fo99,Ru87}, \cite{BGP05,BGP07} and \cite{AKN07,Fe16}, respectively (see also \cite[App. B]{FP17} and \cite{FP23}).

\subsection{$\Gamma-$convergence}\label{subsec:Gammaconv}
Consider a sequence $ \lf\{F_N \ri\}_{N \in \N}$ of closed and lower-bounded quadratic forms on a Hilbert space $\Hil$ with  domains $\dom[F_N]$. Let $F_\infty$  on $\dom[F_\infty]$ be another lower-bounded quadratic form on $\Hil$, to be regarded as the candidate limit of the sequence. We denote by $A_N$ ($N\in \N$) and $A_{\infty}$ the self-adjoint operators associated with $F_N$ and $F_\infty$, respectively. In the sequel, the forms $F_N$ ($N \in \N$) and $F_{\infty}$ are implicitly extended (relaxed) to the whole Hilbert space $\Hil$ by setting their values equal to $+\infty$ outside of their respective domains.

Let $\tau$ denote either the strong or weak topology on $\mathscr{H}$. The sequence $ \lf\{F_N \ri\}_{N \in \N}$ is said to \emph{$\Gamma$-converge} to $F_{\infty}$ w.r.t. $\tau$ on $\Hil$ if the following two conditions are fulfilled:
		\begin{enumerate}[(i)]
            \item ($\Gamma-\liminf$) for any sequence $ \lf\{\psi_N\ri\}_{N \in\N} \subset \Hil$ such that $\psi_N \xrightarrow{\tau} \psi_{\infty} \in\mathscr{H}$, there holds
			\begin{equation*}
				F_{\infty}[\psi_{\infty}] \leqslant \liminf_{N \,\to\, \infty} F_{N}[\psi_N]\,;
			\end{equation*}
			\item ($\Gamma-\limsup$) for any $\psi_{\infty} \in \dom[F_{\infty}]$, there exists a sequence $ \lf\{\psi_N\ri\}_{N\in\N} \subset \dom[F_N]$, such that $\psi_N \xrightarrow{\tau} \psi_{\infty}$ and
			\begin{equation*}
				F_{\infty}[\psi_{\infty}] \geqslant \limsup_{N \,\to\, \infty} F_{N}[\psi_N]\,.
			\end{equation*}
		\end{enumerate}
Since $F_{\infty}$ is lower semicontinuous and bounded from below, without loss of generality we may replace condition (ii) with the following weaker requirement, see \cite[Rem. 1.29]{Br02} and \cite[Lemma 4.2(ii)]{DFT94}:
		\begin{enumerate}
			\item[(ii')] for any $\psi_{\infty} \!\in\! \dom(A_{\infty})$, there exists a sequence $ \lf\{\psi_N \ri\}_{N\in\N} \subset \dom[F_N]$, such that $\psi_N \!\xrightarrow{\tau}\! \psi_{\infty}$ and
			\begin{equation}\label{eq:limsup2}
				F_{\infty}[\psi_{\infty}] = \lim_{N\to \infty} F_{N}[\psi_N]\,.
			\end{equation}
		\end{enumerate}

The main result relating $\Gamma$-convergence of quadratic forms and strong convergence of the corresponding resolvent operators is the following (see, \emph{e.g.}, \cite[Thm. 13.6]{DM93} and \cite[Thm. 1]{BdOV14}):

	\begin{theorem}[$\Gamma-$covergence and operator convergence 1]\label{thm:Gares}	\mbox{}	\\
		Assume the quadratic forms $\{F_N\}_{N\in\N}$ and $F_{\infty}$ to be lower semicontinuous and uniformly bounded from below. Then, the following statements are equivalent:
		\begin{enumerate}[(i)]
			\item $\{F_N\}_{N\in\N}$ $\Gamma$-converges to $F_{\infty}$ w.r.t. the strong topology and condition (i) is valid w.r.t. the weak topology;
			\item $ \lf\{F_{N} + \lambda^2 \| \cdot \|^2 \ri\}_{N\in\N}$ $\Gamma$-converges to $F_{\infty} + \lambda^2\,\| \cdot \|^2$ w.r.t. both the strong and weak topology, for some $\lambda \geqslant 0$;
			\item $ \lf\{ A_{N} \ri\}_{N \in \N}$ converges to $A_{\infty}$ in the strong resolvent sense.
		\end{enumerate}
	\end{theorem}
	
To relate $\Gamma-$convergence and uniform resolvent convergence, it is necessary to introduce a notion of compactness \cite{Mo94}. 
A sequence of quadratic forms $ \lf\{F_N \ri\}_{N \in \N}$, uniformly bounded from below by $\lambda_* \in \R$, is called \emph{asymptotically compact} in $\mathscr{H}$, if for any sequence $ \lf\{\psi_N \ri\}_{N \in \N} \subset \Hil$ such that
	\begin{equation}\label{defn:asymcomp}
		\liminf_{N\,\to\, \infty} \lf(F_N[\psi_N] + \lambda_*^2 \, \lf\|\psi_N \ri\|^2 \ri) < \infty\,,
	\end{equation}
there exists a subsequence $ \lf\{\psi_{N_n} \ri\}_{n\in\N}$, which converges strongly in $\mathscr{H}$.
By a simple variation of \cite[Prop. 4]{DeO11}  (see also \cite{KS03}), we obtain the following:

	\begin{theorem}[$\Gamma-$covergence and operator convergence 2]\label{thm:normres} \mbox{}	\\
		Assume the quadratic forms $ \lf\{F_N\ri\}_{N\in\N}$ and $F_{\infty}$ to be lower semicontinuous and uniformly bounded from below. Moreover, assume that
		\begin{enumerate}[(i)]
			\item $ \lf \{F_N \ri \}_{N \in \N}$ $\Gamma-$converges to $F$ w.r.t. both the strong and weak topology;		
			\item $ \lf\{ F_N \ri\}_{N \in \N}$ is asymptotically compact in $\mathscr{H}$;
			\item $(A_{\infty} + \lambda^2)^{-1}$ is compact for some $\lambda > 0$ large enough;
		\end{enumerate}
		then, $ \lf\{A_N\ri\}_{N \in \N}$ converges to $A_{\infty}$ in norm resolvent sense.
	\end{theorem}

\subsection{Measures and distribution of centers}\label{subsec:meas}	

We discuss here some implications of \cref{ass:UU,ass:xixj} for complex measures related to the empiric measure $\mu_N$ introduced in \eqref{eq:Vwconv}. To this purpose we first briefly recall some basic facts on complex measures to fix the notation.

For any complex measure $ \mu $ on the standard Borel $\sigma-$algebra $\Bor(\R^d)$ of $\R^d$, its {variation} $|\mu|$ is given by
	\begin{equation*}
		|\mu|(E)=\sup \sum_{n=1}^{\infty} |\mu(E_n)|\,, \qquad
		\mbox{for all } E \in \Bor(\R^d)\,,
	\end{equation*}
where the supremum is taken over all partitions $\{E_n\}_{n=1}^{\infty}$ of $E$. 
The total variation $ \lf\|\mu \ri\|_{\mathrm{TV}} := |\mu|(\R^d)$ defines a norm on the space of complex measures $\MeC_{\C}(\R^d)$, which makes the latter a Banach space \cite[Prop. 7.16]{Fo99}. 

A sequence of (complex) measures $ \lf\{\mu_N \ri\}_{N \in \N} \subset \MeC_{\C}(\R^d)$ converges in the weak sense to $\mu \in \MeC_{\C}(\R^d) $ (denoted as $ \mu_N \xrightarrow[N \to +\infty]{\mathrm{w}} \mu$), if
	\begin{equation}\label{eq:weakconv}
		 \lim_{N \to \infty} \int_{\R^d}f\, \diff \mu_{N} = \int_{\R^d}f\, \diff \mu,\qquad \mbox{for all } f\in C_{\mathrm{b}}(\R^d)\,,
	\end{equation}
where $C_{\mathrm{b}}(\R^d)$ is the space of all continuous and bounded complex-valued functions on $\R^d$. Finally, we recall that a sequence of measures $ \lf\{\mu_N \ri\}_{N \in \N} \subset \MeC_{\C}(\R^d)$ is uniformly tight if $ \forall \varepsilon > 0$ there exists a compact set $K_\varepsilon \subset \R^d$ such that $|\mu_N|(\R^d \setminus K_\varepsilon) \leqslant \varepsilon$ for all $N \in \N$.
		    
    \begin{lemma}\label{lem:compmes}
		Let \cref{ass:UU,ass:xixj} hold and let $ \lf\{ \nu_N \ri\}_{N \in \N}$ be the sequence of complex measures 
		\begin{equation}\label{eq:rhoN}
			\nu_N = \frac{1}{N} \sum_{j=1}^N  \pp_j \, \delta_{\xv_j}\,,
		\end{equation}
		where $\pp_j \in \C$ are so that there exists a constant $c > 0$ independent of $N$ such that
		\begin{equation}\label{eq:assaj}
			\frac{1}{N} \sum_{j=1}^N|\pp_j|^2\leqslant c\,.
		\end{equation}
		Then, there exists a weakly convergent subsequence $ \lf\{\nu_{N_k} \ri\}_{k\in\N}$, with $\nu_{N_k} \xrightarrow{\mathrm{w}} \nuU \in  \MeC_{\C}(\R^d)$. Moreover, the limit $\nuU$ is absolutely continuous with respect to $\muU$, {\it i.e.}, there exists $ \pp \in L^2(\R^d,\diff\muU) $ such that
		\begin{equation}\label{eq:nupmu}
			\diff\nuU = \pp\,\diff\muU\,.
		\end{equation}
	\end{lemma}
	
	\begin{proof}
		Let us firstly notice that \cref{ass:UU} ensures the uniform tightness of the sequence $ \lf\{\nu_N \ri\}_{N\in\N}$.
		Furthermore, by Cauchy-Schwartz inequality and condition \eqref{eq:assaj}, it follows that
			\begin{equation*}
				\lf\|\nu_N \ri\|_{\mathrm{TV}}^2 \leqslant \left(\tfrac{1}{N} \tx\sum_{j=1}^N |\pp_j|\right)^{2}
				\leqslant \tfrac{1}{N} \tx\sum_{j=1}^N |\pp_j|^2 \leqslant c\,,
			\end{equation*}
		which shows that the sequence $ \lf\{\nu_N \ri\}_{N \in \N}$ is also uniformly bounded in total variation. Then, by Prokhorov theorem \cite[Thm. 1.4.11]{Bo18}, we infer the existence of a weakly convergent subsequence $ \lf\{\nu_{N_k} \ri\}_{k \in \N}$, with $\nu_{N_k} \xrightarrow[k \to + \infty]{\mathrm{w}} \nuU \in \MeC_{\C}(\R^d)$.
		
		Now, consider the linear functionals on $C_{\mathrm{b}}(\R^d)$ defined by
			\bdm
				\mu_{N_k}(f) = \int_{\R^d}f\,\mathrm{d}\mu_{N_k}\,,\qquad \nu_{N_k}(f) = \int_{\R^d}f\,\mathrm{d}\nu_{N_k}\,,\qquad  \muU(f)=\int_{\R^d}f\,\mathrm{d}\muU\,, \qquad \nuU(f)=\int_{\R^d}f\,\mathrm{d}\nuU\,.
			\edm
		Again by Cauchy-Schwartz inequality, for any $k \in \N$ and for all $f \in C_{\mathrm{b}}(\R^d)$,
		\bdm
			|\nu_{N_k}(f)|^2 = \left|\tfrac{1}{N_k} \tx\sum_{j=1}^{N_k} \pp_j\,f(\xv_{j}) \right|^2 \\
			 \leqslant \left(\tfrac{1}{N_k} \tx\sum_{j=1}^{N_k} |\pp_j|^2\right) \left(\tfrac{1}{N_k} \tx\sum_{j=1}^{N_k} |f(\xv_j)|^2\right)
			\leqslant c \int_{\R^d}|f|^2\, \diff \mu_{N_k} 
			= c\; \mu_{N_k}\big(|f|^2\big)\,. 
		\edm
		On the other hand, by weak convergence of $ \lf\{ \nu_{N_k} \ri\}_{k \in \N} $ and $ \lf\{ \mu_{N_k} \ri\}_{k \in \N} $, we deduce that  $|\nu_{N_k}(f)| \to |\nuU(f)|$ and $|\mu_{N_k}(|f|^2)| \to |\muU(|f|^2)|$, as $k \to \infty$, since $|f|^2\in C_{\mathrm{b}}(\R^d)$. So, we can pass to the limit and obtain
		\begin{equation*}
			|\nuU(f)|^2 \leqslant c\; \mu_{\infty}\big(|f|^2\big) = c\,\|f\|_{L^2(\R^d,\, \diff \muU)}^2\,.
		\end{equation*}
		By standard density arguments, this implies that $\nuU$ extends to a bounded linear functional on $L^2(\R^d, \diff\muU)$. Therefore, by Riesz theorem \cite[Thm. 7.17]{Fo99} (see also \cite[Thm. 6.19]{Ru87}), we infer the existence of a function $\pp \in L^2(\R^d, \diff \muU)$ such that
		\begin{equation*}
			\nuU(f) = \int_{\R^d} f\,p\, \diff \muU,\qquad \mbox{for all } f\!\in\! L^2(\R^d, \diff \muU)\,.
		\end{equation*}
		Finally, since $\UU$ has compact support by \cref{ass:UU}, we also have $\pp\in L^1(\R^d,\diff \muU)$, which proves that $\nuU$ is absolutely continuous w.r.t. $\muU$.
	\end{proof}

Under the assumptions of \cref{lem:compmes}, by linearity, the sequence of complex conjugate measures $\lf\{\overline{\nu_N} \ri\}_{N \in \N}$ given by
		\begin{equation*}
			\overline{\nu_N} = \frac{1}{N} \sum_{j=1}^N \overline{\pp_j}\, \delta_{\xv_j}\,,
		\end{equation*}
admits a subsequence which weakly converges to $\overline{\nuU}$. Of course, the latter fulfills $\diff \overline{\nuU} = \overline{p}\, \diff \mu_\UU$, where $\pp \in L^2(\R^d,\diff \muU)$ is the function identified in \eqref{eq:nupmu}.

\subsection{Green functions}\label{subsec:G0prop}
Consider a Schr\"odinger operator of the form \eqref{eq:H0def} and let $\G0(\xv,\yv)$, with $\lambda > 0$ and $(\xv,\yv) \in \R^d \!\times \R^d$, be the integral kernel associated to its resolvent. Some key properties of $\G0(\xv,\yv)$ were already discussed in \cref{rem:wellQNHN}. Here we report some additional information on pointwise bounds and on the singular behavior of $\G0(\xv,\yv)$ near the diagonal $\{ \xv = \yv\}$.

Standard results on Carleman kernels ensure that $\G0(\xv,\yv)$ is jointly continuous away from the diagonal \cite[Thm. 21]{BGP07}:
	\begin{equation*}
		\G0 \in C^0(\R^d \!\times\! \R^d \setminus \{\xv = \yv\})\,.
	\end{equation*}
Futher properties of $\G0(\xv,\yv)$ are conveniently formulated by comparison with the (real-valued) Green function $\Gf0(\xv-\yv)$ for the free Laplacian:
	    \begin{equation}\label{eq:freegreen}
    			\Gf0(\xv) = \begin{cases}
    				\tfrac{1}{2\pi} K_0\big(\lambda|\xv|\big)\,, & \mbox{if } d=2\,, \vspace{0.1cm} \\
	    			\tfrac{e^{- \lambda|\xv|}}{4\pi|\xv|}\,, & \mbox{if } d=3\,.
			\end{cases}    		 
		\end{equation}
Here $K_0$ denotes the modified Bessel function of the second kind and order zero (\emph{a.k.a.} Macdonald function) \cite[Ch. 10]{OLBC10}.

Exploiting the positivity of the electrostatic potential $\VV$ and the diamagnetic inequality, one infers the following pointwise estimate \cite[Lemma 15]{BGP07}.
	
	\begin{lemma}\label{lemma:G0Gfineq}
		For any $\lambda > 0$ there holds
		\begin{equation}
			|\G0(\xv,\yv)| \leqslant \Gf0(\xv-\yv) \qquad \mbox{for all } (\xv,\yv) \in \R^d \!\times \R^d \setminus \{\xv = \yv\}\,.
		\end{equation}
	\end{lemma}
	
On the other hand, $\G0(\xv,\yv) $ shares the same leading order singularity as $\Gf0(\xv-\yv)$ for $|\xv - \yv| \to 0$. A more refined asymptotic expansion can be derived using the standard Laplace-transform representation
	\begin{equation*}
		\G0(\xv,\yv)  = \int_0^{\infty} e^{-t \lambda^2} \K0(t;\xv,\yv) \,\diff t\,,
	\end{equation*}
where $\K0(t;\xv,\yv)$ is the integral kernel associated to the heat operator $e^{-t \H0}$. Combining this with the small-time asymptotics of $\K0(t;\xv,\yv) $ for $t \to 0^+$ \cite{BK13}, one infers the following:

	\begin{lemma}
		\label{lemma: G regularity}
		For any fixed $\lambda > 0$ there holds
			\begin{equation}\label{eq:G0sing}
				\G0(\xv,\yv)  = \gs(\xv-\yv)\,e^{i \teA(\xv,\yv) } + \hh(\xv,\yv) \,,
			\end{equation}
		where
			\begin{eqnarray}
				& \gs(\xv) := \begin{cases}
						- \frac{1}{2\pi} \log|\xv|\,, & \text{if } d=2\,,\vspace{0.1cm}\\
						\frac{1}{4\pi|\xv|}\,, & \text{if } d=3\,,
						\end{cases} 
						\label{eq:gs} \\
				& \displaystyle{\teA(\xv,\yv)  := (\xv-\yv) \cdot\! \int_0^1\! \AA\big(\yv + s(\xv-\yv)\big)\, \diff s\,.}
						\label{eq:teA}
			\end{eqnarray}
		and $\hh(\xv,\yv) $ is a jointly continuous function on $\R^d \!\times \R^d$.
	\end{lemma}

	\begin{remark}[Local singularity of $\G0(\xv,\yv) $]
	The term $\gs(\xv-\yv)$ in \eqref{eq:G0sing} captures a universal singular behavior on the diagonal. The magnetic phase $\teA \in C^\infty(\R^d \times \R^d)$ is essential when $d = 3$ and $\AA \neq \bm{0}$. In fact, it accounts for a subleading contribution proportional to $(\xv-\yv)/|\xv-\yv|$ in the asymptotic expansion, which would otherwise spoil the continuity of the reminder $\hh(\xv,\yv) $ across the diagonal. In all other cases ({\em i.e.}, $d = 2$ for any $\AA$, or $d = 3$ with $\AA = \bm{0}$), the reminder is continuous regardless of the phase correction.
	\end{remark}
	
\subsection{Hilbert scales associated to $\H0$}\label{subsec:Hilsca}
Exploiting the positivity and self-adjointness of $\H0$, we can define the fractional powers $(\H0 + \one)^{r/2}$ for any $r \in \R$, and, equipping the domain $\dom\big((\H0 + \one)^{r/2}\big)$ with the inner product
	\begin{equation}\label{eq:innr}
		\braketr{\psi_1}{\psi_2}_r := \braketr{(\H0 + \one)^{r/2}\psi_1}{(\H0 + \one)^{r/2}\psi_2},
	\end{equation}
	we can construct the scale of Hilbert spaces
	\begin{equation*}
		\HH^r := \overline{\dom\big((\H0+\one)^{r/2}\big)}^{\,\|\cdot \|_{\HH^r}} , \qquad \mbox{for } r \in \R\,,
	\end{equation*}
where $\|\cdot\|_{\HH^r}$ indicates the norm induced by the modified inner product \eqref{eq:innr}. Then, we have the following property \cite[Prop. 2.16 and Prop. 2.18]{Fe16}.
	\begin{lemma}\label{lem:HHdual}
		For any $r > 0$, $
			\HH^{-r} = (\HH^r)'\,$.
	\end{lemma}
Next, combining local elliptic regularity for positive even orders with complex interpolation theory and standard Sobolev embedding theorem \cite[p. 98, Prop. 2.13]{CP82}, one obtains the following regularity properties \cite[Prop. 2.42 and Cor. 2.43]{Fe16} (see also \cite{FP23}).
	
	\begin{lemma}\label{lem:HHemb}
		The following inclusions are continuous embeddings:
		\begin{enumerate}[(i)]
		\item $\HH^r \subset H^{r}_{\mathrm{loc}}(\R^d)$ for every $r \geqslant 0$;
		\item $\HH^r \subset C_0(\R^d)$ for every $r > d/2$.
		\end{enumerate}
	\end{lemma}
We conclude by recalling the following result \cite[Prop. 2.12]{Fe16}.

	\begin{lemma}\label{lem:Rzext}
		For any $\lambda > 0$, the resolvent $\Rz$ can be uniquely extended to a bounded operator
		\begin{equation}
			\Rz : \HH^r \to \HH^{r+2}\,, \qquad \mbox{for any } r>0\,.
		\end{equation}
	\end{lemma}

\section{Proofs}\label{sec:proofs}

For later convenience let us modify the domain expression \eqref{eq:QNdef} of the quadratic forms $\QN$ according to \eqref{eq:qjres} and use the charges $ \lf\{ p_j = N\, q_j\ri\}_{j \in \lf\{1, \ldots, N \ri\}} $ instead of $ \lf\{ q_j \ri\}_{j \in \lf\{1, \ldots, N \ri\}} $, {\it i.e.},
\begin{eqnarray}
		\dom[\QN] & := & \lf\{\psi = \phi_\lambda + \tfrac{1}{N} \tx\sum_{j=1}^N p_j\, \G0(\,\cdot\,, \xv_j) \in L^2(\R^d) \;\big|\; \phi_{\lambda}\!\in\! \dom[\Q0],\;\, p_1,...\,,p_N \!\in\! \C \ri\} \,, \nonumber  \\
		\QN[\psi] & := & \Q0 \lf[\phi_{\lambda}\ri] + \lambda^2 \lf\|\phi_{\lambda} \ri\|_{2}^2 - \lambda^2 \lf\|\psi\ri\|_{2}^{2} + \tfrac{1}{N^2} \tx\sum_{i,j=1}^N  \overline{p_i}\, \big[\GaN\big]_{ij}\, p_j\,. \nonumber
	\end{eqnarray} 
Note that the operator domain \eqref{eq:HNdef} may be changed accordingly. Let also $\Qinf$ be the alleged $\Gamma$-limit of $ \QN $, as defined in \eqref{eq:Qinf}. In view of \cref{thm:Gares}, it is enough to prove the analogue of \cref{thm:main} for the {shifted} quadratic forms 
			\begin{gather}
				\QN^\lambda := \QN + \lambda^2 \,\|\cdot\|_{2}^2\,, \qquad \dom[\QN^\lambda] = \dom[\QN]\,, \label{eq:lambdaNform} \\
				Q^\lambda_\infty := \Qinf + \lambda^2\, \|\cdot\|_{2}^2\,, \qquad \dom[Q^\lambda_\infty] = \dom[\Qinf]\,,\label{eq:lambdalimitform}
			\end{gather}
where $\lambda > 0$ is a fixed spectral parameter, independent of $N$, to be chosen large enough. As before, we extend these forms to the whole Hilbert space $L^2(\R^d)$ by setting them equal to $+\infty$ outside of their respective domains.

In the sequel, we establish the $ \Gamma-\liminf $ and $ \Gamma-\limsup $ inequalities for the sequence $\{\QN^\lambda\}_{N \in \N}$, which correspond to items (i) and (ii) in \cref{thm:main}. These ultimately account for the $\Gamma$-convergence result reported in \cref{thm:main}. We will finally derive \cref{cor:main} exploiting asymptotic compactness and \cref{thm:normres}.

\subsection{Proof of the $ \Gamma-\liminf $ inequality}
		In this section, we establish item (i) in \cref{thm:main}. To this avail, we begin by collecting a few preliminary auxiliary lemmas.

The following result provides a generalization of some arguments, originally presented in \cite{PV80}.
	
	\begin{lemma}\label{lem:consingm}
		Suppose \cref{ass:UU,ass:xixj} hold and consider the sequence of complex measures $ \lf\{\nu_N \ri\}_{N \in \N}$  \eqref{eq:rhoN}, fulfilling \eqref{eq:assaj} and converging weakly to $\nuU \in \MeC_{\C}(\R^d)$. Moreover, let $\gs$ be defined as in \eqref{eq:gs} and $\xi$ be a Lipschitz continuous function on $\R^d \!\times \R^d$. Then,
		\begin{equation}\label{eq:lemconv}
			\lim_{N \to \infty} \int_{(\R^d\times\R^d) \setminus \Delta}\!\! \gs(\xv-\yv)\,\xi(\xv,\yv)\,\diff \overline{\nu_N}(\xv)\,d\nu_N(\yv)
			= \int_{\R^d \times \R^d}\!\! \gs(\xv-\yv)\,\xi(\xv,\yv)\,\diff\overline{\nuU}(\xv)\,\diff\nuU(\yv)\,,
		\end{equation}
		where $\Delta = \lf\{ (\xv,\yv) \in \R^d \!\times\! \R^d \: \big| \: \xv = \yv \ri\} $.
	\end{lemma}
	
	\begin{proof}
		For the sake of clarity, we first prove the result with $\xi = 1$. We will comment on the minor adjustments needed to handle a generic Lipschitz continuous function $\xi$ at the end of this proof.
		
		For each $j \in \{1,...\,,N\}$, let $B_{N,j}$ be the open ball of center $\yv_j$ and radius $r_N = \frac{\ell}{2}\,N^{-1/d}$ (recall \eqref{eq:infxixj}). Due to \cref{ass:xixj}, for any $i \neq j$, we have $B_{N,i} \cap B_{N,j} = \varnothing$ and the function $y \mapsto \gs(\yv_i - \yv)$ is harmonic on $B_{N,j}$. Applying the mean value property twice, we obtain
			\begin{equation*}
				\gs(\yv_i - \yv_j) = \frac{1}{|B_{N,i}|\,|B_{N,j}|}\int_{B_{N,i} \times B_{N,j}}\! \gs(\xv- \yv)\,\diff \xv\, \diff \yv \,, \qquad \mbox{for all } i \neq j\,.
			\end{equation*} 
		Using the explicit expression \eqref{eq:rhoN} for $\nu_N$ and the above identity, we get
			\bml{\label{eq:proofDelta}
				\int_{(\R^d\times\R^d) \setminus \Delta} \gs(\xv-\yv)\,\diff \overline{\nu_N}(\xv)\,\diff\nu_N(\yv)
				= \frac{1}{N^2}\sum_{\substack{i,j=1,\\ i\neq j}}^N\,\gs(\xv_i - \xv_j)\,\overline{\pp_i}\,\pp_j \\
			 = \frac{1}{N^2}\sum_{\substack{i,j=1,\\ i\neq j}}^N\,\frac{\overline{\pp_i}\,\pp_j}{|B_{N,i}|\,|B_{N,j}|} \int_{B_{N,i} \times B_{N,j}}\! \gs(\xv-\yv)\,\diff \xv\, \diff \yv
				= \int_{(\R^d \times \R^d)\setminus \Delta} \gs(\xv-\yv) \,\overline{\eta_N}(\xv)\,\eta_N(\yv)\,\diff \xv\, \diff \yv, 
			}
		where
			\begin{equation*}
				\eta_N := \frac{1}{N} \sum_{j = 1}^N\, \pp_j\, \frac{\one_{B_{N,j}}}{|B_{N,j}|}\,.
			\end{equation*}
		It can be checked by a straightforward computation that $\eta_N\, \diff x = \diff (\nu_N * \omega_N)$, where $\omega_N$ is the uniform measure on the ball with center at the origin and radius $r_N$. By assumption and direct inspection,
		\bdm
			\nu_N \xrightarrow[N \to +\infty]{\mathrm{w}} \nuU\,,	\qquad	\omega_N \xrightarrow[N \to +\infty]{\mathrm{w}} \delta_{\bm{0}}\,. 
		\edm
		Then, since the convolution of two weakly convergent sequences of measures also weakly converges to the convolution of the limit measures \cite[Ex. 4.8.49]{Bo18}, we get that $\nu_N * \omega_N \xrightarrow{\mathrm{w}} \nuU * \delta_0 = \nuU$, or, equivalently,
			\begin{equation}
				\label{eq: measures conv}
				\lim_{N \to \infty} \int_{(\R^d\times\R^d) \setminus \Delta}\! f(\xv,\yv)\,\diff \overline{\nu_N}(\xv)\, \diff \nu_N(\yv)
				= \int_{\R^d\times\R^d} f(\xv,\yv)\,\diff\overline{\nuU}(\xv)\,\diff\nuU(\yv)\,, \qquad
				\forall f \in C_b(\R^d \!\times \!\R^d) \,.
			\end{equation}
			
		Next, we introduce a continuous and bounded regularization of the singular function $ \gs $. Let us first discuss the three-dimensional case. Setting, for $ \varepsilon > 0 $,
			\begin{equation*}
				\gs^\varepsilon(\xv) := \frac{1}{4\pi (|\xv| + \varepsilon)}\,,
			\end{equation*}
			we have
			\bml{\label{eq:proof1}
				\int_{(\R^3\times\R^3) \setminus \Delta}  \gs(\xv-\yv)\, \diff \overline{\nu_N}(\xv)\, \diff \nu_N(\yv)
				= \int_{(\R^3\times\R^3) \setminus \Delta}  \gs^{\varepsilon}(\xv-\yv)\, \diff \overline{\nu_N}(\xv)\, \diff \nu_N(\yv)
					\\
				+ \int_{(\R^3\times\R^3) \setminus \Delta}  \lf[\gs^{\varepsilon}(\xv-\yv)- \gs(\xv-\yv) \ri] \diff \overline{\nu_N}(\xv)\, \diff \nu_N(\yv)\,. 
			}
		However, since $\gs^{\varepsilon}(\xv - \yv) \in C_b(\R^3 \times \R^3)$, by \eqref{eq: measures conv},
			\begin{equation}
				\lim_{N \to \infty} \int_{\R^3\times\R^3 \setminus \Delta} \gs^{\varepsilon}(\xv-\yv)\, \diff \overline{\nu_N}(\xv)\, \diff \nu_N(\yv)
				= \int_{\R^3\times\R^3} \gs^{\varepsilon}(\xv-\yv)\, \diff \overline{\nuU}(\xv)\, \diff \nuU(\yv) \,.
				\label{eq:proof2}
			\end{equation}
			On the other hand, the sequence $ \lf\{\eta_N \ri\}_{N \in \N}$ is uniformly bounded in $L^2(\R^3)$: recalling \eqref{eq:assaj}, 
			\begin{align*}
				\lf\|\eta_N \ri\|_2^2
				= \frac{1}{N^2} \sum_{j = 1}^N \frac{|\pp_j|^2}{|B_{N,j}|}
				= \frac{6}{\pi \ell^3}\; \frac{1}{N} \sum_{j = 1}^N |\pp_j|^2
				\leqslant \frac{6c}{\pi \ell^3}\,.
			\end{align*}
		Furthermore, by \cref{ass:UU}, $ \supp(\eta_N) $ is contained inside a fixed bounded domain independent of $ N $, which implies that the sequence $ \lf\{\eta_N \ri\}_{N \in \N}$ is uniformly bounded in $L^p(\R^3)$, for any $1 \leqslant p \leqslant 2$. Then, by using the Hardy-Littlewood-Sobolev inequality \cite[Thm. 4.3]{LL01},
			\bmln{
				\lf| \int_{(\R^3\times\R^3) \setminus \Delta}\!\!  \lf[\gs^{\varepsilon}(\xv-\yv) - \gs(\xv-\yv) \ri]\, \diff \overline{\nu_N}(\xv)\, \diff\nu_N(\yv) \ri| 
				\leqslant \frac{\varepsilon}{4\pi}\! \int_{\R^3\times\R^3} \frac{1}{|\xv- \yv|^2}\,|\eta_N(\xv)|\,|\eta_N(\yv)|\,\diff \xv\, \diff \yv	\\
				\leqslant C \varepsilon\,\|\eta_N\|_{L^{2/3}(\R^3)}^2 \leqslant C' \varepsilon \xrightarrow[\eps \to 0^+]{} 0\,,
			}
		where $C, C'> 0$ are suitable constants independent of $N$ and $\varepsilon$.
		It can be shown in a similar fashion that
			\begin{equation}
				\lim_{\varepsilon \to 0^+} \int_{(\R^3\times\R^3) \setminus \Delta} \lf[\gs^{\varepsilon}(\xv-\yv) - \gs(\xv-\yv) \ri]\,d\overline{\nuU}(\xv)\,d\nuU(\yv) = 0\,. \label{eq:proof4}
			\end{equation}	
		Combining \eqref{eq:proof1}-\eqref{eq:proof4} ultimately proves the thesis \eqref{eq:lemconv} for $\xi = 1$ and $d = 3$.
			
		The analogous result with $\xi = 1$ and $d = 2$ can be derived retracing the same arguments, considering the regularized version of $g$ defined as
			\begin{equation*}
				\gs^{\varepsilon}(\xv) := -\tfrac{1}{2\pi} \log(|\xv|+\varepsilon) \qquad (\varepsilon > 0)\,.
			\end{equation*}
		Strictly speaking, such a function is continuous but not bounded, as it diverges at infinity. However, this poses no real issue, since all the measures involved have supports contained within a fixed compact set.

		We conclude by outlining the modifications required for a generic Lipschitz continuous function $\xi$. In this case, a straightforward telescopic argument shows that the analogue of \eqref{eq:proofDelta} becomes
			\bdm
				\int_{(\R^d\times\R^d) \setminus \Delta}\!\! \gs(\xv-\yv) \,\xi(\xv,\yv)\,\diff \overline{\nu_N}(\xv)\, \diff\nu_N(\yv)
				 = \int_{(\R^d \times \R^d)\setminus \Delta}\!\! \gs(\xv-\yv) \,\xi(\xv,\yv)\,\overline{\eta_N}(\xv)\,\eta_N(\yv)\,\diff \xv\, \diff \yv + \Theta_N \,, 
			\edm
		where
			\begin{equation*}
				\Theta_N := \frac{1}{N^2}\sum_{\substack{i,j=1,\\ i\neq j}}^N\,\frac{\overline{\pp_i}\,\pp_j}{|B_{N,i}|\,|B_{N,j}|} \int_{B_{N,i} \times B_{N,j}} \gs(\xv-\yv)\,[\xi(\xv_i,\xv_j) - \xi(\xv,\yv)]\,\diff \xv\, \diff \yv\,.
			\end{equation*}
		By arguments similar to those outlined in the first part of the proof, we readily obtain
			\begin{equation*}
				\lim_{N \to \infty}\int_{(\R^d \times \R^d)\setminus \Delta}\!\! \gs(\xv-\yv) \,\xi(\xv,\yv)\,\overline{\eta_N}(\xv)\,\eta_N(\yv)\,\diff \xv\, \diff \yv
				= \int_{\R^3\times\R^3} \gs(\xv-\yv) \,\xi(\xv,\yv)\,\diff\overline{\nuU(\xv)}\,\diff\nuU(\yv)\,.
			\end{equation*}
		Concerning the reminder $\Theta_N$, using the Lipschitz continuity of $\xi$ and the fact that $|\xv_i-\xv|,|\xv_j - \yv| \leqslant \frac{\ell}{2} N^{-1/d}$ for $\xv \in B_{N,i}$ and $\yv \in B_{N,j}$, we have
			\bdm
				|\Theta_N| 
				\leqslant \frac{\ell}{\sqrt{2}\,N^{1/d}}\,\frac{\|\xi\|_{C^{0,1}}}{N^2}\sum_{\substack{i,j=1,\\ i\neq j}}^N\,\frac{|\pp_i|\,|\pp_j|}{|B_{N,i}|\,|B_{N,j}|} \int_{B_{N,i} \times B_{N,j}} \gs(\xv-\yv)\,\diff \xv\, \diff \yv
				\leqslant \frac{C}{N^{1/d}}\,,
			\edm
		for some constant $C > 0$, which shows that $\Theta_N \to 0$ as $N \to \infty$ and thus concludes the proof.
	\end{proof}

Let us now pass to discuss the positivity of the matrix $\GaN$. To this purpose, let us first report a result proven in \cite[\S 2.4, pp. 25-27 and Remark 2.19]{Se15}. 

	\begin{lemma}\label{lem:rieszen}
		Let \cref{ass:UU,ass:xixj} hold. Then, for any fixed $s \in (0,d)$, there exists a constant $c_s > 0$ such that
		\begin{equation}\label{eq:riesz}
			\limsup_{N\to \infty}\, \frac{1}{N^2}\sum_{\substack{i,j=1\\i\neq j}}^N \frac{1}{|\xv_i - \xv_j|^s}
			\leqslant \int_{\R^d \times \R^d} \frac{1}{|\xv-\yv|^s}\; \diff \muU(\xv)\, \diff \muU(\yv)
			\leqslant c_s\,.
		\end{equation}
	\end{lemma}
	
The arguments presented therein actually establish that the inequalities in \eqref{eq:riesz} remain valid when the function $1/|\xv|^{s}$ is replaced by any $g \in L^1_{\mathrm{loc}}(\R^d) \cap C^0(\R^d \setminus \{\mathbf{0}\})$ which is monotone radial and positive near the origin.

	\begin{lemma}\label{lem:posGamma}
		Let \cref{ass:UU,ass:xixj,ass:a} hold. Then, for any fixed $\lambda > 0$ large enough, there exist $N_\lambda > 0$ and a constant $\gamma_\lambda > 0$, independent of $N$, such that
		\begin{equation}\label{eq:Gammainf}
	    		\tfrac{1}{N}\,\GaN \geqslant \gamma_\lambda\,\one\,, \qquad \mbox{for all } N \geqslant N_\lambda\,.
	    	\end{equation} 
	\end{lemma}
    
    \begin{proof}
		Let us recall the definition \eqref{eq:GaN} of $\GaN$ and set $ \GaN = N(A_N^{\lambda} + B_N^{\lambda}) $, where
			\beq
				\big[A_N^{\lambda}\big]_{ij} := \frac{1}{N}\! \left(\alpha_j + \displaystyle{\lim_{\xv \to \xv_j}}\! \lf( G_0^{\lambda_0}(\xv,\xv_j) - \G0(\xv,\xv_j) \ri)\!\right) \delta_{ij}\,, \qquad
					\big[B_N^{\lambda}\big]_{ij} := -\frac{1}{N}\,\G0(\xv_i,\xv_j)\,(1-\delta_{ij})\,.
	    		\eeq
		Taking into account \cref{ass:a} and the fact that the difference $G_0^{\lambda_0}(\xv,\yv) - \G0(\xv,\yv)$ is jointly continuous on $\R^d \times \R^d$ for all $\lambda,\lambda_0 > 0$ (see \cref{rem:wellQNHN}), it is easy to see that
		\begin{equation*}
			A_N^{\lambda} \geqslant (\min a)\,\one + o(1)\,, \qquad \mbox{as $N \to \infty$}\,.
		\end{equation*}
		Then, the thesis follows as soon as we can establish that the operator norm of $B_N^{\lambda}$ can be made arbitrarily small, uniformly with respect to $N$, by picking $\lambda$ sufficiently large. To this purpose, it is convenient to consider the Hilbert-Schmidt norm, noting that
		\begin{equation*}
			\lf\|B_N^{\lambda}\ri\|^2 \leqslant \lf\|B_N^{\lambda}\ri\|_{\mathrm{HS}}^2 = \frac{1}{N^2} \sum_{\substack{i,j=1\\i\neq j}} \big| \G0(\xv_i,\xv_j) \big|^2 \leqslant \frac{1}{N^2} \sum_{\substack{i,j=1\\i\neq j}} \lf| \Gf0(\xv_i-\xv_j) \ri|^2 ,
		\end{equation*}
		where the last inequality follows from \cref{lemma:G0Gfineq}. Let us also recall the explicit expression \eqref{eq:freegreen} for the free Green function $\Gf0$ and proceed to discuss separately the cases $d = 3$ and $d = 2$.
		
		For $d = 3$, using the elementary inequality $e^{-t} \leqslant t^{-\beta}$, for $t > 0$ and $\beta \in (0,1)$, together with \cref{lem:rieszen} for $s = 2 + \beta$, we obtain
		\begin{equation*}
			\big\|B_N^{\lambda}\big\|^2 
			\leqslant \frac{1}{N^2} \sum_{\substack{i,j=1\\i\neq j}} \left|\frac{e^{- \lambda|\xv_i-\xv_j|}}{4\pi |\xv_i-\xv_j|}\right|^2
			\leqslant \frac{(2\lambda)^{-\beta}}{16\pi^2 N^2} \sum_{\substack{i,j=1\\i\neq j}} \frac{1}{|\xv_i-\xv_j|^{2+ \beta}}
			\leqslant C \lambda^{-\beta} .
		\end{equation*}
		
		For $d = 2$, using a well-known integral representation of the Bessel function $K_0$ \cite[Eq. 10.32.10]{OLBC10}, alongside with the inequality $e^{-t} \leqslant t^{-\beta}$ and \cref{lem:rieszen}, now with $s = 4\beta$ and $\beta \in (0,1/2)$, we get
		\bmln{
			\big\|B_N^{\lambda}\big\|^2 
			\leqslant \frac{1}{16 \pi^2 N^2} \sum_{\substack{i,j=1\\i\neq j}} \left| K_0\big(\lambda|\xv_i-\xv_j|\big)\right|^2
			= \frac{1}{16\pi^2 N^2} \sum_{\substack{i,j=1\\i\neq j}} \left|\int_0^{\infty} \exp\left({-t - \tfrac{\lambda^2|\xv_i-\xv_j|^2}{4t}}\right) \frac{\diff t}{t} \right|^2 \\
			 = \frac{\lambda^{-4\beta}}{2^{4(1-\beta)} \pi^2} \left( \int_0^{\infty} t^{\beta - 1} e^{-t} dt \right)^2 \frac{1}{N^2} \sum_{\substack{i,j=1\\i\neq j}} \frac{1}{|\xv_i-\xv_j|^{4\beta}}
			\leqslant C\, \lambda^{-4\beta} .
		}
		
		Choosing $\lambda > 0$ large enough, the above arguments imply \eqref{eq:Gammainf} for both $d = 3$ and $d = 2$.
    \end{proof}

We are now ready to prove the limit inferior inequality, namely, claim (i) in \cref{thm:main}.
	
	\begin{proposition}[$ \Gamma-\liminf $] \label{prop:liminf}
		\mbox{}	\\
		Let \cref{ass:UU,ass:xixj,ass:a} hold and let $\lambda > 0$ be fixed large enough. Then, for any sequence $ \lf\{\psi_N \ri\}_{N \in \N} \subset L^2(\R^d)$ such that  $ \psi_N \xrightarrow[N \to +\infty]{\mathrm{w}} \psi_\infty$, there holds
		\begin{equation}\label{eq:liminfQN}
			Q^\lambda_\infty[\psi_\infty] \leqslant \liminf_{N \,\to\, \infty} \QN^{\lambda}[\psi_N]\,.
		\end{equation}
	\end{proposition}

	\begin{proof}
		The statement is trivial if the limit inferior on the right-hand side of \eqref{eq:liminfQN} is infinite. Hence, we  assume that $ \psi_N \in \dom[\QN] $. In fact, possibly passing to a subsequence, we can also assume that
        \begin{equation}
        	\label{eq: finite energy}
            \liminf_{N\,\to\,\infty}\QN^{\lambda}[\psi_N] 
            = \lim_{N\,\to\,\infty}\QN^{\lambda}[\psi_N]  \leqslant c\,,
        \end{equation}
        for some finite constant $c > 0$ independent of $ N $. To avoid a too heavy notation however we still label the subsequence by $ N \in \N $.

        In view of \eqref{eq:QNdef}, \eqref{eq:lambdaNform} and \cref{lem:posGamma}, it appears that each term in $\QN^{\lambda}[\psi_N]$ is positive definite for any sufficiently large $N$. Then, the uniform boundedness of $\QN^{\lambda}[\psi_N]$ yields
		\begin{equation}
			\Q0[\phi_{N,\lambda}] + \lambda^2\, \|\phi_{N,\lambda}\|_{2}^2 \leqslant c \,, \qquad
			\tfrac{1}{N^2}\tx\sum_{i,j=1}^N  \overline{\pp_i}\, \big[\GaN\big]_{ij}\, \pp_j \leqslant c\,. \label{eq:unibound}
		\end{equation}
		Since $\Q0 + \lambda^2\|\,\cdot\,\|_2^2$ is a squared norm on $\dom[\Q0]$, the first condition in \eqref{eq:unibound} ensures that, up to extraction of a subsequence, $\{\phi_{N,\lambda}\}_{N \in \N}$ converges weakly in $\dom[\Q0]$ to  $\phi_{\infty,\lambda}$. In other words, by a diagonal argument, we have that the following weak convergences in $L^2(\R^d)$ hold simultaneously:
		\begin{equation*}
			\phi_{N,\lambda} \xrightarrow[N \to + \infty]{\mathrm{w}} \phi_{\infty,\lambda}\,, \qquad
			(-i\nabla + \AA)\phi_{N,\lambda}\xrightarrow[N \to + \infty]{\mathrm{w}} (-i\nabla + \AA)\phi_{\infty,\lambda}\,, \qquad
			\sqrt{\VV}\phi_{N,\lambda} \xrightarrow[N \to + \infty]{\mathrm{w}}\sqrt{\VV}\phi_{\infty,\lambda}\,.
		\end{equation*}
			
		On the other hand, \cref{lem:posGamma} together with the second condition in \eqref{eq:unibound} imply that
		\begin{equation*}
			\tfrac{1}{N} \tx\sum_{j=1}^N |\pp_j|^2 \leqslant \tfrac{c}{\gamma_\lambda} \,.
		\end{equation*}
		We can now apply \cref{lem:compmes}, to deduce that the sequence of complex measures
		\begin{equation*}
			\nu_N = \frac{1}{N} \sum_{j=1}^N  \pp_j \, \delta_{\xv_j} \,,
		\end{equation*}
		admits a weakly convergent subsequence. Furthermore, the limit measure $\nuU \in \MeC_{\C}(\R^d)$ is absolutely continuous with respect to the density $\muU$ introduced in \cref{ass:UU}, and
		\begin{equation*}
			\diff\nuU = \pp\, \diff\muU\,, \qquad \mbox{for some } \pp \in L^2(\R^d,\diff\muU)\,.
		\end{equation*}
		
			We now observe  that $\dom(\H0) = \HH^2 \subset H^2_{\mathrm{loc}}(\R^d) \hookrightarrow C^0(\R^d)$ (see \cref{rem:wellQNHN} and \cref{lem:HHemb}). Furthermore, considering that the supports of the measures $\nu_N$ and $\nu_\infty$ are all contained within a fixed bounded set, we have that  $\nu_N \in (\HH^2)' \equiv \HH^{-2}$, for both $d = 2$ and $d = 3$ (see \cref{lem:HHdual}). Hence, by  \cref{lem:Rzext}, $\Rz \nu_N \in L^2(\R^d)$ and, for any $f \in L^2(\R^d)$, 				\bdm
			\braketl{\Rz \nu_N}{f} = \frac{1}{N} \sum_{j=1}^N  \overline{\pp_j} \,(\Rz f)(\xv_j) \\
			= \frac{1}{N} \sum_{j=1}^N \overline{\pp_j} \displaystyle{\int_{\R^d}} \G0(\xv_j, \yv)\, f(\yv)\, \diff \yv
			= \braketl{\tfrac{1}{N} \tx\sum_{j=1}^N \pp_j\, \G0(\,\cdot\,, \xv_j)}{f}\,,
		\edm
		which shows that
		\begin{equation}\label{eq:RznuN}
			\lf( \Rz \nu_N \ri)(\xv) = \frac{1}{N} \sum_{j=1}^N \pp_j\, \G0(\xv\,,\xv_j) \,.
		\end{equation}
		In addition, from the weak convergence of $\nu_{N}$ it follows that 
		\begin{equation}\label{eq:wlimproof}
			\lim_{N \to+\infty} \braketl{\Rz \nu_N}{f}
			=\braketl{\Rz \nuU}{f},
		\end{equation}
		which, by the same arguments outlined above and \cref{ass:a}, implies that
		\begin{align*}
	 		\lf( \Rz \nuU \ri)(\xv)
	 		= \int_{\R^d}\! \G0(\xv\,,\yv)\,\pp(\yv)\,\UU(\yv)\, \diff \yv
	 		= \lf(\Rz\, \UU \pp\ri)(\xv) \,.
		\end{align*}
		Notice that, since $U \in L^\infty(\R^d)$ has compact support and $\pp \in L^2(\R^d,\UU\,\diff \xv)$, we have indeed $\pp \in L^2(\R^d)$, $\UU \pp \in L^2(\R^d)$ and $\Rz\,\UU \pp \in \HH^2 = \dom(\H0) \subset \dom[\Q0]$.
		Summing up, for any weakly convergent sequence $ \lf\{ \psi_N \ri\}_{N \in \N} $ satisfying \eqref{eq: finite energy}, we get 
		\bdm
			\psi_N  = \phi_{N,\lambda} + \frac{1}{N}\sum_{j=1}^N \pp_j\, \G0(\, \cdot \,, \xv_j)
			\;\xrightarrow[N \,\to \,\infty]{\mathrm{w}} \;
			\psi_\infty = \phi_{\infty,\lambda} + \Rz\,\UU \pp  \in \dom[\Q0]\,.	
		\edm

		Recalling the definition of $Q^\lambda_\infty$ in \eqref{eq:Qinf} and \eqref{eq:lambdalimitform}, by direct calculations, we further obtain
		\begin{align*}
			Q^\lambda_\infty[\psi_\infty] & = \lf\|(-i \nabla + \AA) \psi_{\infty} \ri\|_2^2 + \lf\|\sqrt{\VV}\psi_{\infty} \ri\|_2^2 + \lambda^2 \lf\|\psi_\infty\ri\|_2^2 - \meanlrlr{\psi_{\infty}}{\tfrac{\UU}{a}}{\psi_{\infty}} \\
			& = \|(-i \nabla + \AA) \phi_{\infty,\lambda}\|_2^2 
				+ \| \sqrt{\VV} \phi_{\infty,\lambda}\|_2^2 
				+ \lambda^2\,\|\phi_{\infty,\lambda}\|_2^2 \\
			& \qquad 
				+ 2 \Re \big\langle \phi_{\infty,\lambda}, (\H0 + \lambda^2) \Rz\,\UU \pp \big\rangle 
				+ \langle \Rz\, \UU \pp, (\H0 + \lambda^2)\, \Rz\, \UU \pp \rangle 
				- \left\langle\psi_{\infty},\tfrac{\UU}{a}\,\psi_{\infty}\right\rangle \\
			& = \Q0[\phi_{\infty,\lambda}] + \lambda^2\,\lf\|\phi_{\infty,\lambda}\ri\|_2^2
				- \braketr{ \UU \pp}{ \Rz\,\UU \pp }
				+ \braketr{ \pp}{a\, \UU \pp }
				- \Big\| \sqrt{a\, \UU} \pp - \sqrt{\tfrac{\UU}{a}}\,\psi_{\infty}\Big\|_2^2\,.
		\end{align*}		
		Now, we can write
		\bml{
			\QN^\lambda[\psi_N] - Q^\lambda_\infty[\psi_\infty] = \lf( \Q0[\phi_{N,\lambda}] + \lambda^2 \lf\|\phi_{N,\lambda} \ri\|_{2}^2 \,\ri)
				- \lf( \Q0[ \phi_{\infty,\lambda}] + \lambda^2 \lf\|\phi_{\infty,\lambda}\ri\|_2^2 \,\ri)
				    \\ 
				     + \tfrac{1}{N^2}\tx\sum_{i \neq j}  \overline{\pp_i}\, \big[\GaN\big]_{ij}\, \pp_j
				+ \braketr{\UU \pp}{ \Rz\,\UU \pp}	 + \tfrac{1}{N^2}\tx\sum_{j=1}^N  \big[\GaN\big]_{jj}\, |\pp_j|^2
				\\
				- \braketr{\pp}{ a\,\UU \pp}
				+ \Big\| \sqrt{a\, \UU} \pp - \sqrt{\tfrac{\UU}{a}}\,\psi_{\infty}\Big\|_2^2 \,. \label{eq:diffQinfQN}
		}
		Let us discuss separately the addenda appearing on the right-hand side of \eqref{eq:diffQinfQN}.
		Firstly, from the weak convergence of $\{\phi_{N,\lambda}\}_{N \in \N}$ to $\phi_{\infty,\lambda}$ in $\dom[\Q0]$ and the weak lower-semicontinuity of the norm, we readily get
		\begin{equation*}
			\liminf_{N \,\to\,\infty} \lf( \Q0[\phi_{N,\lambda}] + \lambda^2 \lf\|\phi_{N,\lambda} \ri\|_{2}^2 \ri)
				\geqslant \Q0[\phi_{\infty,\lambda}] + \lambda^2 \lf\|\phi_{\infty,\lambda} \ri\|_2^2 \,.
		\end{equation*}
		Secondly, recalling \eqref{eq:G0sing} and noting that $\teA(\xv,\yv) \in C^\infty(\R^d \times \R^d)$ and $\hh(\xv,\yv) \in C^0(\R^d\times \R^d)$, by \cref{lem:consingm} and the weak convergence of $\{\nu_N\}_{N \in \N}$ to $\nuU$ in $\MeC_{\C}(\R^d)$, we infer
		\bml{
			\lim_{N \to +\infty} \tfrac{1}{N^2} {\tx\sum_{i \neq j}}\, \overline{\pp_i}\, \big[\GaN\big]_{ij}\, \pp_j
			= - \lim_{N \to +\infty} \tfrac{1}{N^2} {\tx\sum_{i \neq j}}\,  \overline{\pp_i}\, \lf[ \gs(\xv_i-\xv_j)\,e^{i \teA(\xv_i,\xv_j)} + \hh(\xv_i,\xv_j) \ri] \, \pp_j  \\
			= - \int_{\R^d \times \R^d} \big[\gs(\xv -\yv)\,e^{i \teA(\xv,\yv)} + \hh(\xv,\yv)\big]\,\diff\overline{\nuU}(\xv)\,\diff\nuU(\yv) = -\, \braketr{\UU \pp}{\Rz\, \UU \pp} . \label{eq:convGamma}
		}
		Thirdly, notice that the weak convergence of the complex measures $\nu_N$, together with the continuity and boundedness of $a$ (see \cref{ass:a}), implies
		\begin{equation*}
			 \lim_{N \,\to\, + \infty} \frac{1}{N} \sum_{j = 1}^N \overline{p_j} \sqrt{a(\xv_j)}\,\varphi(\xv_j) =  \int_{\R^d} \sqrt{a}\,\varphi \, \diff \nuU\,,
			\qquad \mbox{for any } \varphi \in C_b(\R^d)\,.
		\end{equation*}
		By Cauchy-Schwartz inequality and the weak convergence of $ \lf\{ \mu_N \ri\}_{N \in \N} $ to $ \mu_\infty$ (see \cref{ass:UU}), it follows that
		\begin{equation*}
			\lf| \braketl{\sqrt{a}\,\pp}{\varphi}_{L^2(\R^d,\,\diff \muU)} \ri|
			\leqslant \left(\lim_{N \,\to\, \infty} \tfrac{1}{N} \tx\sum_{j = 1}^N a(\xv_j)\, |p_j|^2 \right)^{1/2} \lf\|\varphi\ri\|_{L^2(\R^d,\,\diff\muU)}\,.
		\end{equation*}
		Hence,
		\bml{
			\label{eq: 2}
			\braketr{\pp}{\,a\,\UU \pp} = \lf\|\sqrt{a}\,\pp \ri\|_{L^2(\R^d,\, \diff \muU)}^2
			= \sup_{0 \,\neq\, \varphi \in C_b(\R^d)} \frac{\lf| \braketl{\sqrt{a}\,\pp}{\varphi}_{L^2(\R^d,\, \diff \muU)} \ri|^2}{\lf\|\varphi\ri\|_{L^2(\R^d,\,\diff \muU)}} \leqslant \lim_{N \,\to\, \infty} \tfrac{1}{N} \tx\sum_{j = 1}^N a(\xv_j)\, |p_j|^2   \\
			= \disp\lim_{N \,\to\, \infty} \tfrac{1}{N^2} {\tx \sum_{j = 1}^N} \alpha_j\, |p_j|^2
			\leqslant \disp\lim_{N \,\to\, \infty} \tfrac{1}{N^2} \tx\sum_{j = 1}^N \big[\GaN\big]_{jj}\, |p_j|^2 \,,
		}
		by \cref{ass:a} and the explicit form of the diagonal entries of the matrix $\GaN$.
		
		Now, dropping the last addendum on the right-hand side of \eqref{eq:diffQinfQN} thanks to its positivity and exploiting that $\liminf_{N\to \infty}(s_N+t_N)\geqslant \liminf_{N\to \infty} s_N+\liminf_{N\to \infty}t_N$, we combine \eqref{eq:diffQinfQN}-\eqref{eq: 2}, to obtain that
		\begin{equation}
			\liminf_{N \,\to\, \infty}\QN^\lambda[\psi_N] - Q^\lambda_\infty[\psi_\infty] \geqslant 0\,,
		\end{equation}
		which finally proves the claim \eqref{eq:liminfQN}.
	\end{proof}

\subsection{Proof of the $ \Gamma-\limsup $ (in)equality}
	We now turn to the proof of item (ii) in \cref{thm:main}. In this connection, we recall that it suffices to prove the claim for $\psi_{\infty}\in\dom(\Hinf)$, where $\Hinf$ is the selfadjoint operator associated to $\Qinf$ (see condition (ii') in \cref{subsec:Gammaconv} and the related comments).
\medskip

	First of all, keeping in mind that $\dom(\Hinf) \equiv \dom(\H0)$  and recalling the identity \eqref{eq:psiap}, we proceed to construct a so-called {\em recovery sequence} $ \lf\{\psi_N\ri\}_{N \in \N}$ for any given $\psi_{\infty}\in \dom(\H0)$.
	
	\begin{lemma}\label{lem:recovery}
		Let \cref{ass:UU,ass:a} hold. Let $\psi_{\infty} \in \dom(\H0)$ and consider the sequence of functions given by
		\begin{equation}\label{eq:psiNdef}
			\psi_N(\xv) := \phi_{\lambda}(\xv) +\frac{1}{N} \sum_{j=1}^N \pp(\xv_j)\,\G0(\xv\,, \xv_j)\,, 
		\end{equation}
		where
		\beq
			\pp := \tfrac{1}{a}\,\psi_{\infty} \in C^0(\R^d)\,, 
			\qquad \mbox{and}\qquad
			\phi_{\lambda} := \psi_{\infty} - \Rz\, \UU \pp \in \dom(\H0)\,.
			\label{eq:pjphidef}
		\eeq
		Then, the following properties hold:
		\begin{enumerate}[(i)]
			\item $\psi_N\in \dom[\QN^\lambda]$, for all $N \in \N$;
			\item $ \lf\{\psi_N \ri\}_{N \in \N}$ converges in norm to $\psi_\infty$, as $N\to \infty$.
		\end{enumerate}
	\end{lemma}
	\begin{proof}
		{\em (i)} Since $\psi_{\infty} \in \dom(\H0) = \HH^2\subset C^0(\R^d)$ by \cref{lem:HHemb}, \cref{ass:a} implies that  $p = \psi_\infty/a$ is a continuous function on $\R^d$. So, the coefficients $p(\xv_j) \in \C$ appearing in \eqref{eq:psiNdef} are well-defined. To say more, considering that $\UU \in L^\infty(\R^d)$ has compact support by \cref{ass:UU}, we have that $\UU p \in L^2(\R^d)$ and, as a consequence, $\phi_{\lambda} = \psi_{\infty} - \Rz\, \UU p \in \dom(\H0) \subset \dom[\Q0]$. In view of \eqref{eq:QNdef} and \eqref{eq:lambdaNform}, the above arguments suffice to infer that $\psi_N\in \dom[\QN^\lambda]$ for any $N\in\N$.

		{\em (ii)} By \eqref{eq:psiNdef} and \eqref{eq:pjphidef},
		\begin{equation*}
			\psi_N(\xv) - \psi_\infty(\xv) = \tfrac{1}{N} \tx\sum_{j=1}^N \pp(\xv_j)\,\G0(\xv\,, \xv_j) - \lf(\Rz\, \UU \pp \ri) (\xv) \,.
		\end{equation*}
		We can then reproduce the argument in the proof of \cref{prop:liminf} (see \eqref{eq:wlimproof} and the related comments) to show that the right-hand side of the above expression converges weakly to $0$ in $L^2(\R^d)$. Then, to conclude the proof we just have to prove that
		\begin{equation}\label{eq:slimeq}
			\lim_{N \to \infty} \lf\|\tfrac{1}{N} \tx\sum_{j=1}^N \pp(\xv_j)\,\G0(\,\cdot\,,\xv_j) \ri\|_2 = \lf\|\Rz\, \UU \pp\ri\|_2\,,
		\end{equation}
		but
		\bdm
			\lf\|\tfrac{1}{N} \tx\sum_{j=1}^N \pp(\xv_j)\,\G0(\,\cdot\,, \xv_j) \ri\|_2^2
			= \tfrac{1}{N^2} {\tx\sum_{i,j=1}^N}\, \overline{\pp}(\xv_i)\, \pp(\xv_j) \int_{\R^d} \G0(\xv_i, \yv)\,\G0(\yv, \xv_j)\, \diff \yv\,
		\edm
		and the integral in the above expression actually coincides with the integral kernel of ${\Rz}^2$ evaluated at $(\xv_i,\xv_j)$. Considering that this kernel is indeed a bounded and continuous function on $\R^d \times \R^d$ \cite[Thm. 21, item (5)]{BGP07}, and that $p$ is continuous on $\R^d$, by weak convergence of product measures \cite[Prop. 2.7.8]{Bo18}, we readily obtain 
		\bmln{
			\lim_{N \to \infty} \lf\|\tfrac{1}{N} {\tx\sum_{j=1}^N} \pp(x_j)\,\G0(\,\cdot\,, \xv_j) \ri\|_2^2
			= \int_{\R^d \times \R^d} \overline{\pp}(\xv)\, \pp(\yv) \left(\int_{\R^d} \G0(\xv, \yv')\,\G0(\yv', \yv)\,\diff \yv'\right) \UU(\xv)\,\UU(\yv)\, \diff \xv \diff \yv \\
			 = \braketr{ \UU \pp}{ {\Rz}^2 \UU \pp}
			= \lf\|\Rz \,\UU \pp \ri\|_2^2\,,
		}
		which proves \eqref{eq:slimeq}, whence the thesis.
    \end{proof}

    We are now in a position to establish the limit superior identity.
     	
	\begin{proposition}\label{prop:limsup}
		\mbox{}		\\
		Let \cref{ass:UU,ass:xixj,ass:a} hold and let $\lambda > 0$. Then, for any $\psi_\infty \in \dom(\H0)$ there exists a sequence $ \lf\{\psi_N \ri\}_{N \in \N}$ such that $\psi_N \in \dom[\QN^\lambda]$ for all $N \in \N$, $\psi_N \xrightarrow[N \to + \infty]{} \psi_\infty$, and
		\begin{equation}\label{eq:limsupQN}
			\lim_{N \,\to\, \infty} \QN^{\lambda}[\psi_N] = Q^\lambda_\infty[\psi_\infty]\,.
		\end{equation}
	\end{proposition}
		
	\begin{proof}
		Let us consider the recovery sequence $ \lf\{\psi_N \ri\}_{N \in \N}$ introduced in \cref{lem:recovery}. We only have to establish the identity \eqref{eq:limsupQN}. 
		
		Noting that $\psi_\infty, \Rz\, \UU \pp \in \dom(\H0)$, integrating by parts and using the basic identity $(\H0 + \lambda^2) \Rz\, \UU \pp = \UU \pp \in L^2(\R^d)$, we obtain
		\bmln{
			\QN^\lambda[\psi_N] - Q^\lambda_\infty[\infty]
			= \big\| (-i \nabla + \AA) (\psi_{\infty}\! - \Rz\, \UU \pp) \big\|_{2}^{2} 
				+ \big\| \sqrt{\VV}(\psi_{\infty}\! - \Rz\, \UU \pp) \big\|_{2}^{2} 
				+ \lambda^2 \big\|\psi_{\infty}\! - \Rz\, \UU \pp \big\|_{2}^2  \\
				+ \tfrac{1}{N^2}\tx\sum_{i,j=1}^N  \overline{p(\xv_i)}\, \big[\GaN\big]_{ij}\, p(\xv_j)
					 - \left[ \lf\|(- i \nabla + \AA)\psi_\infty \ri\|_{2}^2 + \big\|\sqrt{\VV}\psi_\infty \big\|_{2}^2 + \lambda^2 \lf\|\psi_\infty \ri\|_{2}^2 - \meanlrlr{\psi_\infty}{\tfrac{\UU}{a}}{\psi_\infty} \right]
			\\
			= 	- \braketr{\pp}{a\, \UU \pp}
				+ \braketr{\UU \pp}{\Rz\,\UU \pp} 
				+ \tfrac{1}{N^2}\tx\sum_{j=1}^N  \big[\GaN\big]_{jj}\, |p(\xv_j)|^2
				+ \tfrac{1}{N^2}\tx\sum_{i \neq j} \overline{p(\xv_i)}\, \big[\GaN\big]_{ij}\, p(\xv_j)\,.
		}
		
		Since both $a$ and $p$ are continuous by \cref{ass:a} and \eqref{eq:pjphidef}, the weak convergence of measures in \cref{ass:UU} yields
		\bdm
			\lim_{N \,\to\, \infty} \tfrac{1}{N^2}{\tx\sum_{j = 1}^N \big[\GaN\big]_{jj}\, |p(\xv_j)|^2}
				= \lim_{N \,\to\, \infty} \tfrac{1}{N^2} {\tx\sum_{j = 1}^N \big(N a(\xv_j) + \mathcal{O}(1) \big)\, |p(\xv_j)|^2}
				= \int_{\R^d} a \,|p |^2\, \diff \muU
				= \braketr{p}{a\,\UU p} .
		\edm
		On the other hand, by arguments analogous to those reported in the proof of \cref{prop:liminf} (see, in particular, \eqref{eq:convGamma} and the related comments), we infer that
		\bdm
			\lim_{N \to \infty} \frac{1}{N^2}{\sum_{\substack{i,j=1,\\i\neq j}}^N \overline{p(\xv_i)}\, \big[\GaN\big]_{ij}\, p(\xv_j)}
			= - \braketr{\UU \pp}{ \Rz\,\UU \pp}  ,
		\edm
		which completes the proof.
    \end{proof}

\subsection{Proof of the uniform resolvent convergence}
In this conclusive section we present the proof of \cref{cor:main}. Also in this case, we firstly state an auxiliary result on the compactness properties of the extended resolvent operator $\Rz$ introduced in  \cref{lem:Rzext}.
	
	\begin{lemma}\label{lem:compres}
	Let $\lambda > 0$. If the resolvent $\Rz$ is compact on $L^2(\R^d)$, then it extends to a compact operator
		\begin{equation}
			\Rz : \HH^r \to \HH^{r+2-\varepsilon}\,, \qquad \mbox{for all }\, r \in \R\,,\; \varepsilon > 0\,.
		\end{equation}
	\end{lemma}
	
	\begin{proof}
		By the compactness of $\Rz$ as an operator acting in $L^2(\R^d)$, we infer that $\H0$ has purely discrete spectrum $\sigma(\H0) = \lf\{\lambda_\ell^2\ri\}_{\ell \in \N}$, with $\lambda_\ell \geqslant \lambda_{\ell +1}$ and $\lambda_\ell \to \infty$ as $\ell \to \infty$ \cite[Prop. 5.12]{Sc12}. In particular, we have the identity
		\begin{equation*}
			\Rz \psi = \tx\sum_{\ell \in \N} \frac{1}{\lambda_\ell^2 + \lambda^2}\, \braket{b_\ell}{\psi}\,b_\ell\,,
		\end{equation*}
	where $\{b_\ell\}_{\ell \in \N}$ is an orthonormal basis of eigenvectors of $\H0$. It is easy to check that $b_\ell \in \HH^r$ for any $r \in \R$.
	Let us now introduce the sequence of truncated operators
		\begin{equation*}
			R^\lambda_L \psi = \tx\sum_{\ell \leqslant L} \frac{1}{\lambda_\ell^2 + \lambda^2}\, \braket{b_\ell}{\psi}\,b_\ell\,\,, \qquad \mbox{for } L \in \N\,.
		\end{equation*}
	These operators have finite-rank, so they are compact. To infer the thesis, it suffices to show that $R^\lambda_L - \Rz \to 0$ as $L \to \infty$, with respect to the topology of bounded operators from $\HH^r$ to $\HH^{r+2-\varepsilon}$, for $\varepsilon > 0$. In this connection, let us point out that
		\bmln{
			\lf\| (\Rz - R^\lambda_L) \psi \ri\|_{\HH^{r + 2 - \varepsilon}}
			= \left\|\tx\sum_{\ell \, \geqslant\, L} \frac{(\lambda_\ell^2 + 1)^{r/2+1-\varepsilon/2}}{\lambda_\ell^2 + \lambda^2}\,\braket{b_\ell}{\psi}\,b_\ell \right\|_2 \\
			\leqslant \tfrac{1}{(\lambda_L^2 + \lambda^2)^\varepsilon}\, \lf\|\tx\sum_{\ell\, \geqslant\, L} (\lambda_\ell^2 + \one)^{r/2}\,\,\braket{b_\ell}{\psi}\,b_\ell \ri\|_2
			\leqslant \tfrac{1}{(\lambda_L^2 + \lambda^2)^\varepsilon}\, \lf\|\psi \ri\|_{\HH^r}.
		}
	 Recalling that $\lambda_L \to \infty$ for $L \to \infty$, this implies that
		\begin{equation*}
			\lf\|\Rz - R^\lambda_L \ri\|_{\mathscr{B}(\HH^{r},\HH^{r + 2 - \varepsilon})} \leqslant \tfrac{1}{(\lambda_L^2 + 1)^\varepsilon} \;\xrightarrow[~L \,\to \,\infty~]\; 0\,,
		\end{equation*}
	ultimately proving the thesis.
	\end{proof}

	We now prove the claim stated in \cref{cor:main}. 
	
    \begin{proof}[Proof of \cref{cor:main}] 
		The first claim follows by \cref{thm:main} and standard arguments of $\Gamma-$convergence theory \cite[\S 13]{DM93}. The lifting of the convergence to norm resolvent sense, whenever the resolvent $\Rz = (\H0 + \lambda^2)^{-1}$ is compact for some $ \lambda > 0  $ is typically taken from granted, but it is in fact a delicate question in general. Indeed, very often, one can exploit a uniform control on the sequence of resolvents $ (\HN + \lambda^2)^{-1} $, like {\it e.g.} when the domains of the operators or of the quadratic forms are independent of $ N $, and approximate them by finite rank operators, which allows to deduce convergence in norm. This, however, is certainly not obvious when, like in our case, the form domains depend on $ N $. Therefore, we provide here a detailed proof of the statement.
		
		 The second claim indeed follows from \cref{thm:normres}, once we assess asymptotic compactness of the forms $\QN$. Namely, we have to show that for any sequence $\lf\{\psi_N \ri\}_{N \in \N} \subset L^2(\R^d)$ satisfying 
		\begin{equation}\label{eq:asycomhyp}
			\liminf_{N \to \infty}\left(\QN[\psi_N] + \lambda^2 \lf\|\psi_N\ri\|_{2}^2\right)< +\infty\,,
		\end{equation}
	for $\lambda > 0$ sufficiently large, one can extract a subsequence which converges strongly in $L^2(\R^d)$ (see \eqref{defn:asymcomp}), that will still be label by $N \in \N$.
	
    Arguing as in the proof of \cref{prop:liminf}, we may infer that the condition \eqref{eq:asycomhyp} ensures that there exists a finite constant $c > 0$ such that
	    \begin{eqnarray*}
			\Q0[\phi_{N,\lambda}] + \lambda^2\, \|\phi_{N,\lambda}\|_{2}^2 \leqslant c \,, \qquad
			\tfrac{1}{N^2}\tx\sum_{j=1}^N  |\pp_j|^2 \leqslant c\,.
		\end{eqnarray*}
 On the one hand, the compactness of $\Rz$ implies that the operator $(\H0 + \lambda^2)^{-1/2}$ is compact, as well. Considering that $\|(\H0 + \lambda^2)^{1/2}\psi_N\|_2^2$ yields an equivalent norm on $\dom[\Q0]$, by classical arguments, we get that $\dom[\Q0]$ is compactly embedded in $L^2(\R^d)$ (see, {\it e.g.}, \cite[Prop. 5.12]{Sc12}).
    	Since $\phi_{N,\lambda}$ is uniformly bounded in $\dom[\Q0]$, we may extract a strongly convergent subsequence.
	
	On the other hand, let us refer once more to the sequence of measures
		\begin{equation*}
			\nu_N := \tfrac{1}{N} \tx \sum_{j =1}^N \pp_j\,\delta_{\xv_j}\,, \qquad \mbox{for } N \in \N\,,
		\end{equation*}
	and recall that, by \eqref{eq:RznuN},
		\begin{equation*}
			\tfrac{1}{N}\tx\sum_{j=1}^N q_j\, \G0(\,\cdot\,, \xv_j)	 = \Rz \nu_N\,.
		\end{equation*}
	Thanks \cref{lem:HHemb}, for any $\varepsilon \in (0,2-d/2)$ and $f \in \HH^{2 - \varepsilon} \subset C^0(\R^d)$, we have
    		\bdm
			\lf| \braket{\nu_N}{f} \ri| 
			\leqslant \tfrac{1}{N} {\tx \sum_{j =1}^N} |\pp_j|\,|f(\xv_j)|
			\leqslant \left(\tfrac{1}{N} {\tx \sum_{j =1}^N} |\pp_j|^2\right)^{\!1/2} \sup_{\xv \in \supp\UU} |f(\xv)|
			\leqslant c^{1/2}\, \lf\|f \ri\|_{\HH^{2-\varepsilon}}\,.
    		\edm
	Hence, the sequence $\nu_N$ is uniformly bounded in $(\HH^{2-\varepsilon})' = \HH^{-(2-\varepsilon)}$ (see \cref{lem:HHdual}) and by \cref{lem:compres}, we infer that $\tfrac{1}{N}\tx\sum_{j=1}^N q_j\, \G0(\,\cdot\,, \xv_j)$ converges strongly in $L^2(\R^d)$, up to the extraction of a subsequence.
	
	Putting together the above results, by a diagonal argument, we obtain that $\psi_N$ admits a subsequence which converges strongly in $L^2(\R^d)$, thus concluding the proof.      
    \end{proof}

\end{document}